\documentclass[letter,12pt]{article}

\usepackage{fullpage}

\usepackage{amsthm,amsmath,natbib}

\usepackage{parskip}
\usepackage{amsfonts}
\usepackage{enumerate}
\usepackage{graphicx}
\usepackage{rotating}
\usepackage{natbib}

\newtheorem{theo}{Theorem}[section]

\newtheorem{coro}{Corollary}[section]

\newcommand{\Prob}{\mathbb{P}}
\newcommand{\R}{\mathbb{R}}
\newcommand{\E}{\mathbb{E}}

\newcommand{\abs}[1]{\left| #1 \right|}

\newcommand{\calN}{\mathcal{N}}
\DeclareMathOperator{\Cov}{Cov}

\newcommand{\leftb}{\Big}
\newcommand{\leftsmall}{\big}
\newcommand{\rightb}{\Big}
\newcommand{\rightsmall}{\big}

\newcommand{\footnoteremember}[2]{
  \footnote{#2}
  \newcounter{#1}
  \setcounter{#1}{\value{footnote}}
}
\newcommand{\footnoterecall}[1]{
  \footnotemark[\value{#1}]
}

\begin{document}

\title{P-Values for High-Dimensional Regression}

\author{Nicolai Meinshausen\footnoteremember{same}{These authors
    contributed equally to this work}\footnote{Department of Statistics,
    University of Oxford, UK} \and 
  Lukas Meier\footnoterecall{same}\footnoteremember{zurich}{Seminar f\"ur
    Statistik, ETH Zurich, Switzerland} \and 
  Peter B\"uhlmann\footnoterecall{zurich}}

\maketitle
\begin{abstract}
  Assigning significance in high-dimensional regression is
  challenging. Most computationally efficient selection algorithms cannot
  guard against inclusion of noise variables. Asymptotically valid p-values
  are not available. An exception is a recent proposal by
  \citet{wasserman08highdim} which splits the data into two parts. The
  number of variables is then reduced to a manageable size using the first
  split, while classical variable selection techniques can be applied to
  the remaining variables, using the data from the second split. This
  yields asymptotic error control under minimal conditions. It involves,
  however, a one-time random split of the data. Results are sensitive to
  this arbitrary choice: it amounts to a `p-value lottery' and makes it
  difficult to reproduce results.  Here, we show that inference across
  multiple random splits can be aggregated, while keeping asymptotic
  control over the inclusion of noise variables. We show that the resulting p-values can be used for control of both family-wise error (FWER) and false discovery rate (FDR).
In addition, the proposed
  aggregation is shown to improve power while reducing the number of
  falsely selected variables substantially.
\end{abstract}


\section{Introduction}
The problem of high-dimensional variable selection has received tremendous
attention in the last decade. Sparse estimators like the Lasso
\citep{tibshirani96lasso} and extensions thereof \citep{zou06adalasso,
  meinshausen05relaxedlasso} have been shown to be very powerful because
they are suitable for high-dimensional data sets and because they lead to
sparse, interpretable results.

In the usual work-flow for high-dimensional variable selection problems,
the user sets potential tuning parameters to their prediction optimal
values and uses the resulting estimator as the final result. In the
classical low-dimensional setup, some error control based on p-values is a
widely used standard in all areas of sciences. So far, p-values were not
available in high-dimensional situations, except for the proposal of
\citet{wasserman08highdim}. An ad-hoc solution for assigning relevance is
to use the bootstrap to analyze the stability of the selected predictors
and to focus on those which are selected most often (or even
always). \cite{bach08bootstrap} and \citet{meinshausen08stability}  show for the Lasso that this leads to a
consistent model selection procedure under fewer restrictions than for the
non-bootstrap case.

More recently, some progress has been achieved to obtain error control
\citep{wasserman08highdim, meinshausen08stability}. Here, we build upon the
approach of \cite{wasserman08highdim} and show that an extension of their
`screen and clean' algorithm leads to a more powerful variable selection
procedure. Moreover, family-wise error rate (FWER) and false discovery rate (FDR) can be controlled, while
\cite{wasserman08highdim} focus on variable selection rather than assigning
significance via p-values. We also extend methodology to control of the false discovery rate \citep{benjamini95fdr} for high-dimensional data.

 While the main application of the procedure are high-dimensional data, where the number $p$ of variables can greatly exceed sample size $n$, we show that the method is also quite competitive with more standard error control for $n>p$ settings, indeed often giving a better detection power in the presence of highly correlated variables. 

This article is organized as follows. We discuss the single-split method of
\cite{wasserman08highdim} briefly in Section~\ref{section:datasplitting},
showing that results can strongly depend on the arbitrary choice of a
random sample splitting. We propose a multi-split method, removing this
dependence. In Section~\ref{section:error} we prove FWER and FDR-control of the
multi-split method, and we show in Section~\ref{section:numerical}
numerically for simulated and real data-sets that the method is more
powerful than the single-split version while reducing substantially the
number of false discoveries.
 Some possible extensions of the proposed
methodology are outlined in Section~\ref{section:extensions}.

\section{Sample Splitting and High-Dimensional Variable Selection}
\label{section:datasplitting}
We consider the usual high-dimensional linear regression setup with a
response vector $Y = (Y_1, \ldots, Y_n)$ and an $n \times p$ fixed design
matrix $X$ such that
\[
  Y = X\beta + \varepsilon,
\]
where $\varepsilon = (\varepsilon_1, \ldots \varepsilon_n)$ is a random
error vector with $\varepsilon_i$ iid.\ $\calN(0, \sigma^2)$ and $\beta \in
\R^p$ is the parameter vector. Extensions to other models are outlined in
Section \ref{section:extensions}. 

Denote by
\[
  S = \{j; \, \beta_j \neq 0\}
\]
the set of active predictors and similarly by $N = S^c = \{j; \, \beta_j =
0\}$ the set of noise variables. Our goal is to assign p-values for the
null-hypotheses $H_{0,j}: \beta_j = 0$ versus $H_{A,j}: \beta_j \neq 0$ and
to infer the set $S$ from a set of $n$ observations $(X_i,Y_i)$,
$i=1,\ldots,n$. We allow for potentially high-dimensional designs, i.e.\ $p
\gg n$. This makes statistical inference very challenging.  An approach
proposed by \citet{wasserman08highdim} is to split the data into two parts,
reducing the dimensionality of predictors on one part to a manageable size
of predictors (keeping the important variables with high probability), and
then to assign p-values and making a final selection on the second part of
the data, using classical least squares estimation.

\subsection{FWER control with the Single-Split Method}
The procedure of \cite{wasserman08highdim} attempts to control the family-wise error rate (FWER), which is defined as the probability of making at least one false rejection. The method relies on sample-splitting,
performing variable selection and dimensionality reduction on one part of
the data and classical significance testing on the remaining part. The data
are splitted randomly into two disjoint groups $D_{in} = (X_{in}, Y_{in})$
and $D_{out} = (X_{out}, Y_{out})$ of equal size.  Let $\tilde{S}$ be a
variable selection or screening procedure which estimates the set of active
predictors. Abusing notation slightly, we also denote by $\tilde{S}$ the
set of selected predictors.  Then variable selection and dimensionality
reduction is based on $D_{in}$, i.e.\ we apply $\tilde{S}$ only on
$D_{in}$. This includes the selection of potential tuning parameters
involved in $\tilde{S}$.  The idea is to break down the large number $p$ of
potential predictor variables to a smaller number $k\ll p$ with $k$ at most
a fraction of $n$ while keeping all relevant variables. The regression
coefficients and the corresponding p-values
$\tilde{P}_1,\ldots,\tilde{P}_p$ of the selected predictors are determined
based on $D_{out}$ by using ordinary least squares estimation on the set
$\tilde{S}$ and setting $\tilde{P}_j=1$ for all $j\notin \tilde{S}$. If the
selected model $\tilde{S}$ contains the true model $S$, i.e.\ $\tilde{S}
\supseteq S$, the p-values based on $D_{out}$ are unbiased. Finally, each
p-value $\tilde{P}_j$ is adjusted by a factor $|\tilde{S}|$ to correct for
the multiplicity of the testing problem.

The selected model is given by all variables in $\tilde{S}$ for which the
adjusted p-value is below a cutoff $\alpha\in (0,1)$,
\[ 
  \hat{S}_{\mathit{single}} = \leftb\{ j\in\tilde{S} : \tilde{P}_j |\tilde{S}|
 \le \alpha \rightb\}. 
\] 
Under suitable assumptions discussed later, this yields asymptotic control
against inclusion of variables in $N$ (false positives) in the sense that
\[ 
  \limsup_{n\to\infty} \Prob\leftb[|N\cap \hat{S}_{single}| \geq 1\rightb] \le
  \alpha,
\] 
i.e.\ control of the family-wise error rate. The method is easy to
implement and yields the asymptotic control under weak assumptions.  The
single-split method relies, however, on an arbitrary split into $D_{in}$
and $D_{out}$. Results can change drastically if this split is chosen
differently. This in itself is unsatisfactory since results are not
reproducible.

\subsection{FWER control with the New Multi-Split Method}
An obvious alternative to a single arbitrary split is to divide the sample
repeatedly. For each split we end up with a set of p-values. It is not
obvious, though, how to combine and aggregate the results.

 In the remainder of the section, we will give a
possible answer. For each hypothesis, a distribution of p-values is obtained for random sample splitting. We will propose that error control can be based on the quantiles of this distribution.  We will show empirically that, maybe
unsurprisingly, the resulting
procedure is more powerful than the single-split method. The multi-split method also makes results reproducible, at
least approximately if the number of random splits is chosen to be very
large.

The multi-split method uses the following procedure:

\noindent
For $b = 1, \ldots, B$:
    \begin{enumerate}
    \item Randomly split the original data into two disjoint groups
      $D^{(b)}_{in}$ and $D^{(b)}_{out}$ of equal size.
    \item Using only $D^{(b)}_{in}$, estimate the set of active predictors
      $\tilde{S}^{(b)}$.
    \item 
      \begin{enumerate}
      \item Using only $D^{(b)}_{out}$, fit the selected variables in
        $\tilde{S}^{(b)}$ with ordinary least squares and calculate the
        corresponding p-values $\tilde{P}^{(b)}_j$ for $j \in
        \tilde{S}^{(b)}$.
        \item Set the remaining p-values to 1, i.e.\
          \[
            \tilde{P}^{(b)}_j = 1, \, j \notin \tilde{S}^{(b)}.
          \]
      \end{enumerate}
    \item Define the adjusted (non-aggregated) p-values as
      \begin{equation}\label{eq:pvaladj}
        P^{(b)}_j = \min\leftb(\tilde{P}^{(b)}_j |\tilde{S}^{(b)}|, \,
          1\rightb), \; j = 1, \ldots, p
      \end{equation}
    \end{enumerate}
Finally, aggregate over the $B$ p-values $P_j^{(b)}$, as discussed below.

The procedure leads to a total of $B$ p-values for each predictor
$j=1,\ldots,p$. It will turn out that suitable summary statistics are
quantiles. For $\gamma \in (0, 1)$ define
\begin{equation}
  \label{eq:quant-pval}
  Q_j(\gamma) = \min\leftb\{1, q_{\gamma}\big(\{P^{(b)}_j / \gamma; \, b = 1, \ldots, B\}\big)\rightb\}, 
\end{equation}
where $q_{\gamma}(\cdot)$ is the (empirical) $\gamma$-quantile function. 

A p-value for each predictor $j=1,\ldots,p$ is then given by $Q_j(\gamma)$,
for any fixed $0 < \gamma < 1$. We will show in Section \ref{section:error}
that this is an asymptotically correct p-value, adjusted for
multiplicity. To give an example, for a choice of $\gamma=0.5$, the quantity $Q_j(0.5)$ is twice the median of all p-values $P^{(b)}_j$, $b=1,\ldots,B$.

 A proper selection of $\gamma$ may be
difficult. Error control is not guaranteed anymore if we search for the
best value of $\gamma$.
We propose to use instead an adaptive version which selects a suitable
value of the quantile based on the data.  Let $\gamma_{\min} \in (0, 1)$ be
a lower bound for $\gamma$, typically $0.05$, and define
\begin{equation}
  \label{eq:quant-pval-opt}
  P_j = \min\leftb\{ 1, \big(1 - \log \gamma_{\min} \big)
  \inf_{\gamma \in (\gamma_{\min}, 1)} Q_j(\gamma). \rightb\}
\end{equation}

The extra correction factor $1 - \log \gamma_{\min}$ ensures that the
family-wise error rate remains controlled at level $\alpha$ despite of the
adaptive search for the best quantile, see Section \ref{section:error}. For
the recommended choice of $\gamma_{\min}=0.05$, this factor is upper
bounded by 4; in fact, $1-\log(0.05)\approx 3.996$. 

We comment briefly on the relation between the proposed adjustment to false
discovery rate \citep{benjamini95fdr, benjamini01control} or family-wise
error \citep{holm79simple} controlling procedures. While we provide a
family-wise error control and as such use union bound corrections as in
\cite{holm79simple}, the definition of the adjusted p-values
(\ref{eq:quant-pval-opt}) and its graphical representation in Figure
\ref{fig:hist} are vaguely reminiscent of the false discovery rate
procedure, rejecting hypotheses if and only if the empirical distribution
of p-values crosses a certain linear bound. The empirical distribution in
(\ref{eq:quant-pval-opt}) is only taken for one predictor variable, though,
which is either in $S$ or $N$. This would correspond to a multiple testing
situation where we are testing a single hypothesis with multiple
statistics. 
\begin{figure}
\includegraphics[width=0.45\textwidth]{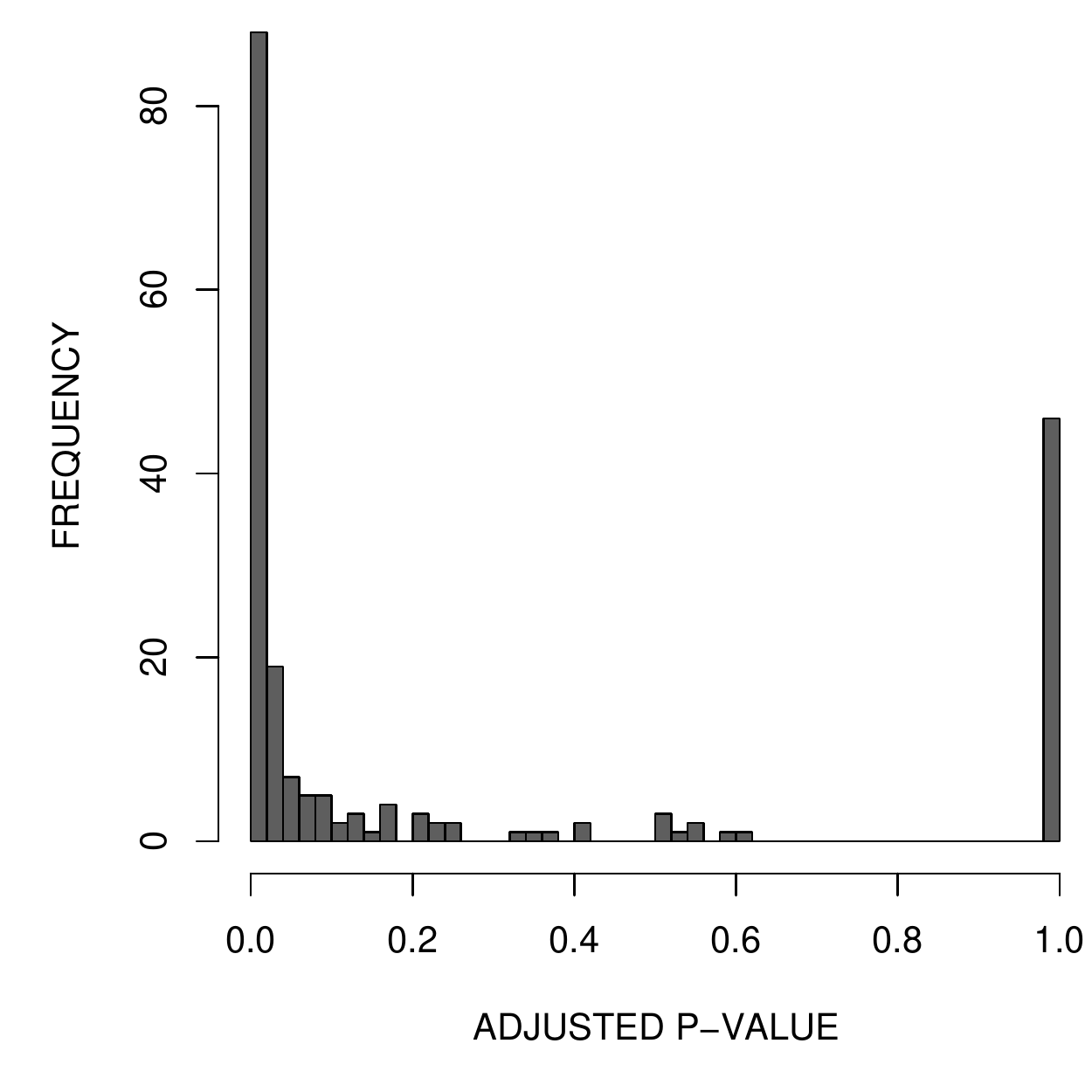}
\includegraphics[width=0.45\textwidth]{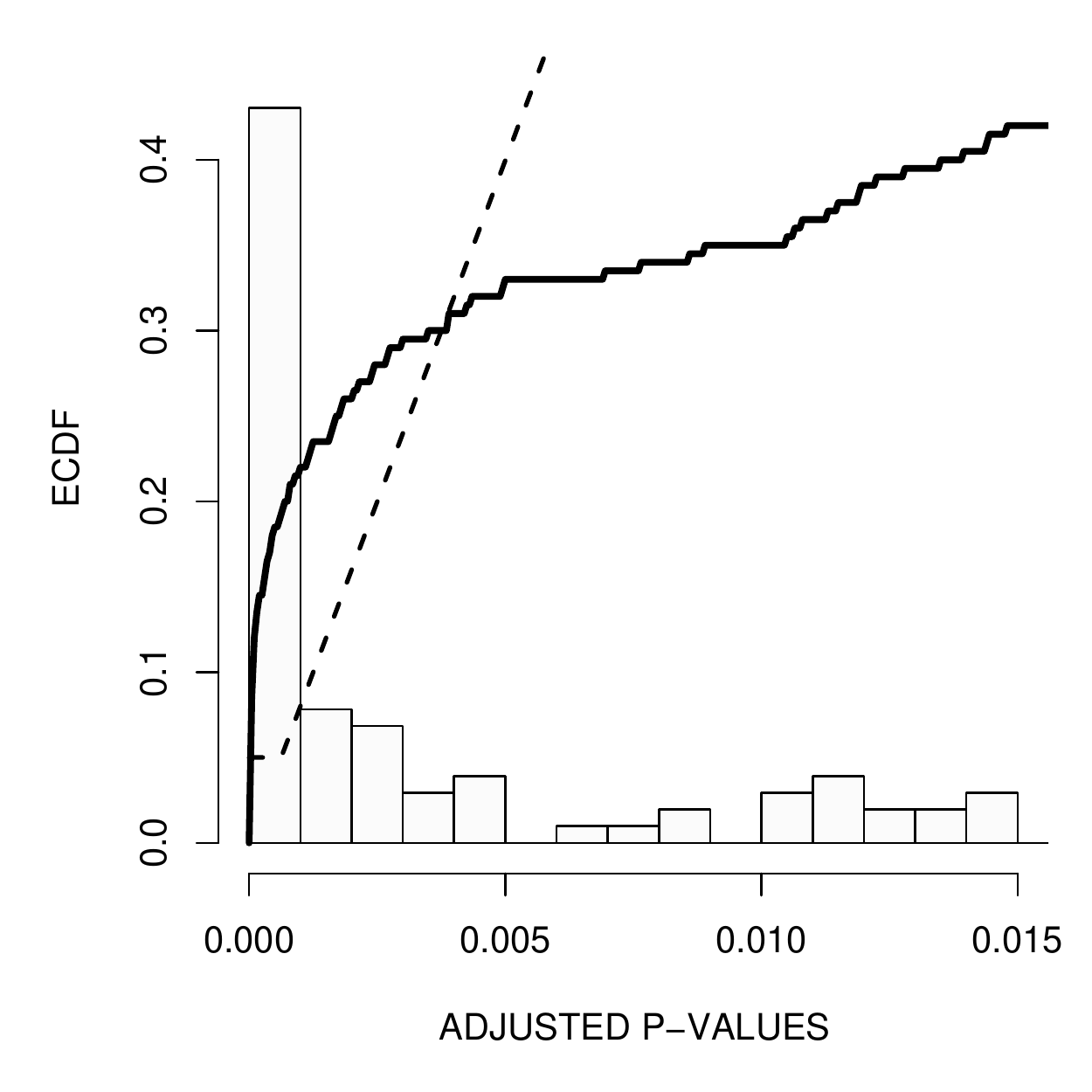}
\caption{\textit{ \label{fig:hist} Left: a histogram of adjusted p-values
    $P^{(b)}_j$ for the selected variable in the motif regression data
    example of Section \ref{sec:realdata-motifregression}. The single split
    method picks randomly one of these p-values (a `p-value lottery') and
    rejects if it is below $\alpha$. For the multi-split method, we reject
    if and only if the empirical distribution function of the adjusted
    p-values crosses the broken line (which is $f(p) = \max\{0.05,
    (3.996/\alpha) p \}$) for some $p\in (0,1)$. This bound is shown as a
    broken line for $\alpha=0.05$. For the given example, the bound is
    indeed exceeded and the variable is thus selected.}}
\end{figure}
Figure~\ref{fig:hist} shows an example. The left panel contains the
histogram of the adjusted p-values $P^{(b)}_j$ for $b=1,\ldots,B$ of the
selected variable in the real data example in Section
\ref{sec:realdata-motifregression}. The single split method is equivalent
to picking one of these p-values randomly and selecting the variable if
this randomly picked p-value is sufficiently small. To avoid this `p-value
lottery', the multi-split method computes the empirical distribution of all
p-values $P^{(b)}_j$ for $b=1,\ldots,B$ and rejects if the empirical
distribution crosses the broken line in the right panel of
Figure~\ref{fig:hist}. A short derivation of the latter is as
follows. Variable $j$ is selected if and only if $P_j\le\alpha$, which
happens if and only if there exists some $\gamma\in(0.05,1)$ such that
$Q_j(\gamma) \le \alpha/(1-\log 0.05) \approx \alpha/3.996$.  Equivalently,
using definition (\ref{eq:quant-pval}), the $\gamma$-quantile of the
adjusted p-values, $q_\gamma(P^{(b)}_j)$, has to be smaller than or equal
to $\alpha\gamma/3.996$. This in turn is equivalent to the event that the
empirical distribution of the adjusted p-values $P^{(b)}_j$ for
$b=1,\ldots,B$ is crossing above the bound $f(p) = \max\{0.05,
(3.996/\alpha) p\}$ for some $p\in (0,1)$. This bound is shown as a broken
line in the right panel of Figure~\ref{fig:hist}.

The resulting adjusted p-values $P_j$, $j=1,\ldots,p$ can then be used for both FWER and FDR control. 
For FWER control at level $\alpha\in (0,1)$, simply all p-values below $\alpha$ are rejected and the selected subset is 
\begin{equation}
  \label{eq:selected-model}
  \hat{S}_{\mathit{multi}} = \{ j  : P_j \le \alpha \}.
\end{equation}
We will show in Section \ref{section:FWERcontrol} that indeed, asymptotically, $\Prob(V>0)\le \alpha$, where $V=|\hat{S}_{\mathit{multi}}\cap N|$ is the number of falsely selected variables under the proposed selection (\ref{eq:selected-model}).
Besides better reproducibility and asymptotic family-wise error control,
the multi-split version is, maybe unsurprisingly, more powerful than the
single-split selection method.

\subsection{FDR control with the multi-split method}\label{section:FDR}
Control of the family-wise error rate is often considered as too conservative.
If many rejections are made, \citet{benjamini95fdr} proposed to control instead the expected proportion of false rejections, the false discovery rate (FDR).
Let $V= |\hat{S} \cap N    | $ be the number of false rejections for a selection method $\hat{S}$ and $R= |\hat{S}|$ the total number of rejections. The false discovery rate is defined as the expected proportion of false rejections
\begin{equation}\label{eq:FDRdef} \E(Q),   \qquad \mbox{with  }\;\;  Q=V/\max\{1,R\} .\end{equation}
For no rejections, $R=0$, the denominator ensures that the false discovery proportion $Q$ is 0, conforming with the definition in  \citet{benjamini95fdr}.

The original FDR controlling procedure in \citep{benjamini95fdr} first orders the observed p-values as $P_{(1)}\le P_{(2)} \le \ldots \le P_{(p)}$ and defines
\begin{equation}\label{kFDR} k = \max\{ i: P_{(i)} \le \frac{i}{p} q \}.\end{equation} Then all variables or hypotheses with the smallest $k$ values are rejected and no rejection is made if the set in (\ref{kFDR}) is empty. FDR is controlled this way at level $q$ under the condition that all p-values are independent. It has been shown in \citet{benjamini01control} that the procedure is conservative under a wider range of dependencies between p-values; see also \citet{blanchard2008tss} for related work. It would, however, require a big leap of faith to assume any such assumption for our setting of high-dimensional regression. For general dependencies, \citet{benjamini01control} showed that control is guaranteed at level  $q \sum_{i=1}^p i^{-1}  \approx q(1/2 + \log (p))$. 

The standard FDR procedure is working with the raw p-values, which are assumed to be uniformly distributed on $[0,1]$  for true null hypotheses. The division by $p$ in (\ref{kFDR}) is an effective correction for multiplicity.
The proposed multi-split method, however, is producing already adjusted p-values, as in (\ref{eq:quant-pval-opt}). 
Since we are working already with multiplicity-corrected p-values, the division by $p$ in (\ref{kFDR}) turns out to be superfluous. Instead, we can order the corrected p-values $P_j$, $j=1,\ldots,p$ in increasing order $P_{(1)}\le P_{(2)}\le \ldots \le P_{(p)}$ and select the $h$ variables with the smallest p-values, where 
\begin{equation}\label{hFDR} h= \max\{ i: P_{(i)} \le i q \}.\end{equation} 
The selected set of variables is denoted, with the value of $h$ given in (\ref{hFDR}), by
\begin{equation}
  \label{eq:selected-model-FDR}
  \hat{S}_{\mathit{multi};\mathit{FDR}} = \{ j : P_j \le P_{(h)} \},
\end{equation}
with no rejections, $\hat{S}_{\mathit{multi};\mathit{FDR}}=\emptyset$, if $P_{(i)} > iq$ for all $i=1,\ldots,p$.

The procedure (\ref{eq:selected-model-FDR}) will achieve FDR control at level $q \sum_{i=1}^p i^{-1} \approx q(1/2+\log p)$. To get FDR control at level $q$, we replace $q$ in (\ref{hFDR}) by $q/(\sum_{i=1}^p i^{-1})$, completely analogous to the standard FDR-procedure under arbitrary dependence of the p-values in \citet{benjamini01control}.
We will prove error control in the following section and show empirically the advantages of the proposed multi-split version over both the single-split and  standard FDR  controlling procedures in the later section with numerical results.

\section{Error Control and Consistency}
\label{section:error}
\subsection{Assumptions}
To achieve asymptotic error control, a few assumptions are made in
\cite{wasserman08highdim} regarding the crucial requirements for the
variable selection procedure $\tilde{S}$.

\begin{enumerate}[({A}1)]
\item \label{screening-property}
  \emph{Screening property}: $\lim_{n \to \infty} \Prob\leftb[\tilde{S}
  \supseteq S\rightb] \; =\;  1.$
\item \label{sparsity-property}
  \emph{Sparsity property:} $|\tilde{S}| < n/2$.
\end{enumerate}
The \emph{screening property} (A\ref{screening-property}) ensures that all
relevant variables are retained. Irrelevant noise variables are allowed to
be selected, too, as long as there are not too many as required by the
\emph{sparsity property} (A\ref{sparsity-property}). A violation of the
sparsity property would make it impossible to apply classical tests on the
retained variables.

The Lasso \citep{tibshirani96lasso} is an important example which satisfies
(A\ref{screening-property}) and (A\ref{sparsity-property}) under
appropriate conditions discussed in \cite{meinshausen06lasso},
\cite{zhao06consistency}, \cite{geer07glmlasso},
\cite{meinshausen06lassorecovery2} and \cite{bickel2008sal}. The adaptive
Lasso \citep{zou06adalasso, zhang07lasso} satisfies
(A\ref{screening-property}) and (A\ref{sparsity-property}) as well under
suitable conditions. Other examples include, assuming appropriate
conditions, $L_2$ Boosting \citep{friedman01boosting,
  buhlmann06highdimboost}, orthogonal matching pursuit \citep{tropp07omp}
or Sure Independence Screening \citep{fan08sis}.

We will typically use the Lasso (and extensions thereof) as screening
method. Other algorithms would be possible. \cite{wasserman08highdim}
studied various scenarios under which these two properties are satisfied
for the Lasso, depending on the choice of the regularization parameter. We
refrain from repeating these and similar arguments, just working on the
assumption that we have a selection procedure $\tilde{S}$ at hand which
satisfies both the \emph{screening property} and the \emph{sparsity
  property}.

\subsection{FWER  control}\label{section:FWERcontrol}
We proposed two versions for multiplicity-adjusted p-values. One is
$Q_j(\gamma)$ as defined in (\ref{eq:quant-pval}) which relies on a choice
of $\gamma\in(0,1)$. The second is the adaptive version $P_j$ defined in
(\ref{eq:quant-pval-opt}) which makes an adaptive choice of $\gamma$. We
show that both quantities are multiplicity-adjusted p-values providing
asymptotic FWER-error control.

\phantom{k} 
\begin{theo}
  \label{theo:fam-error-single}
  Assume (A\ref{screening-property}) and (A\ref{sparsity-property}). Let
  $\alpha, \gamma$ $\in (0, 1)$. If the null-hypothesis $H_{0,j}: \beta_j =
  0$ gets rejected whenever $Q_j(\gamma) \leq \alpha$, the family-wise error rate
  is asymptotically controlled at level $\alpha$, i.e.\
  \[
    \limsup_{n \to \infty} \; \Prob\Big[ \min_{j \in N} Q_j(\gamma) \leq
      \alpha \Big] \leq \alpha,
  \]
  where $\mathbb{P}$ is with respect to the data sample and the statement
  holds for any of the $B$ random sample splits.
\end{theo}
A proof is given in the appendix. 

Theorem \ref{theo:fam-error-single} is valid for any pre-defined value of
the quantile $\gamma$.
However, the adjusted p-values $Q_j(\gamma)$ involve the somehow arbitrary
choice of $\gamma$ which might pose a problem for practical
applications. We therefore proposed the adjusted p-values $P_j$ which
search for the optimal value of $\gamma$ adaptively.  

\phantom{k}
\begin{theo}
  \label{theo:fam-error-all}
  Assume (A\ref{screening-property}) and (A\ref{sparsity-property}). Let
  $\alpha$ $\in (0, 1)$. If the null-hypothesis $H_{0,j}: \beta_j = 0$ gets
  rejected whenever $P_j \leq \alpha$, the family-wise error rate is
  asymptotically controlled at level $\alpha$, i.e.\
  \[
    \limsup_{n \to \infty} \; \Prob\leftb[\min_{j \in N} P_j \leq
    \alpha\rightb] \leq \alpha,
  \]
  where the probability $\mathbb{P}$ is as in Theorem
  \ref{theo:fam-error-single}. 
\end{theo}
A proof is given in the appendix. 

A brief remark regarding the asymptotic nature of the results seems in order. The proposed error control relies on all truly important variables being selected in the screening stage with very high probability. This is our \emph{screening property} (A\ref{screening-property}). Let $\mathcal{A}$ be the event $S\subseteq \tilde{S}$. The results above for example in Theorem \ref{theo:fam-error-all} can be formulated in a non-asymptotic way as 
 $
    \Prob[\mathcal{A} \; \cap \; \{\min_{j \in N} P_j \leq
    \alpha\}] \leq \alpha,
  $
and $P(\mathcal{A})\rightarrow 1$, typically exponentially fast, for $n\rightarrow\infty$. Analogous remarks apply to Theorem \ref{theo:fam-error-single} and \ref{theo:fdr-all} below.

\subsection{FDR control}
The adjusted p-values can be used for FDR control, as laid out in Section \ref{section:FDR}. The set of selected variables $\hat{S}_{\mathit{multi;FDR}}$ was defined in (\ref{eq:selected-model-FDR}). Here, we show that FDR is indeed controlled at the desired rate with this procedure.

\phantom{k}
\begin{theo}
  \label{theo:fdr-all}
  Assume (A\ref{screening-property}) and (A\ref{sparsity-property}). Let
  $q \in (0,1)$. Let $\hat{S}_{\mathit{multi;FDR}}$ be the set of selected variables, as defined in (\ref{eq:selected-model-FDR}) and $V=|\hat{S}_{\mathit{multi;FDR}}\cap N|$ and $R=|\hat{S}_{\mathit{multi;FDR}}|$. The false discovery rate (\ref{eq:FDRdef}) with $Q=V/\max\{1,R\}$ is then 
  asymptotically controlled at level $q\sum_{i=1}^p i^{-1}$, i.e.\
  \[
 \limsup_{n \to \infty}\; \E(Q) \; \le\;  q \sum_{i=1}^p \frac{1}{i}.
  \]
\end{theo}
A proof is given in the appendix.

As with FWER-control, we could be using, for any fixed value of $\gamma$, the values $Q_j(\gamma)$, $j=1,\ldots,p$ instead of $P_j$, $j=1,\ldots,n$. We refrain from giving the full details since, in our experience, the adaptive version above works reliably and does not require an a-priori choice of the quantile $\gamma$ that is necessary otherwise.

\subsection{Model Selection Consistency}
If we let level $\alpha = \alpha_n \to 0$ for $n \to \infty$, the
probability of falsely including a noise variable vanishes because of the
preceding results. In order to get the property of consistent model
selection, we have to analyze the asymptotic behavior of the power. It
turns out that this property is inherited from the single-split method.

\phantom{k}
\begin{coro}
  \label{coro:model-consistency}
  Let $\hat{S}_{\mathit{single}}$ be the selected model of the single-split
  method. Assume that $\alpha_n \to 0$ can be chosen for $n \to \infty$ at
  a rate such that $\lim_{n \to \infty} \Prob[\hat{S}_{\mathit{single}} =
  S] = 1$.  Then, for any $\gamma_{\min}$ (see (\ref{eq:quant-pval-opt})),
  the multi-split method is also model selection consistent for a suitable
  sequence $\alpha_n$, i.e.\ for $\hat{S}_{\mathit{multi}} = \{j \in
  \tilde{S}; P_j \leq \alpha_n\}$ it holds that
  \[
    \lim_{n \to \infty} \Prob\leftb[\hat{S}_{\mathit{multi}} = S\rightb] = 1.
  \]
\end{coro}
\cite{wasserman08highdim} discuss conditions which ensure that $\lim_{n \to
  \infty} \mathbb{P}[\hat{S}_{single} = S] = 1$ for various variable
selection methods such as the Lasso or some forward variable selection
scheme.

The reverse of the Corollary above is not necessarily true. The multi-split
method can be consistent if the single-split method is not. A necessary condition for consistency of the single-split method is $\limsup_{n\to\infty} \Prob[ P_j^{(b)}\le \alpha] =1  $  for all $j\in S$, where the probability is with respect to both the data and the
random split-point, as there is a positive probability otherwise that variable $j$ will not be selected with the single-split approach. For the multi-split method, on the other hand, we only need a bound on quantiles of $P_j^{(b)}$ over $b=1,\ldots,B$.  We refrain from going into more
details here and rather show with numerical results that the multi-split
method is indeed more powerful than the single-split analogue. 
We also remark that the Bonferroni correction in (\ref{eq:pvaladj}), multiplying the raw p-values with the number $|\tilde{S}^{(b)}|$ of selected variables, could possibly be improved upon by using ideas in \citet{hothorn08simultaneous}, further improving the power of the procedure.

\section{Numerical Results}
\label{section:numerical}
In this section we compare the empirical performance of the different
estimators on simulated and real data sets. Simulated data allow a thorough
evaluation of the model selection properties. The real data set shows that
we can find signals in data with our proposed method that would not be
picked up by the single-split method. We use a default value of $\alpha =
0.05$ everywhere.

\subsection{Simulations}
We use the following simulation settings:
\begin{enumerate}[(A)]
  \item Simulated data set with $n = 100$, $p = 100$ and a design matrix
    coming from a centered multivariate normal distribution with covariance
    structure $\Cov(X_j, X_k) = \rho^{\abs{j - k}}$ with $\rho=0.5$. \label{sim-lowdim}
  \item As (\ref{sim-lowdim}) but with $n = 100$ and $p = 1000$. 
    \label{sim-highdim}
  \item \label{sim-dsm} Real data set with $n = 71$ and $p = 4088$ for the
    design matrix $X$ and artificial response~$Y$.
\end{enumerate}
The data set in (\ref{sim-dsm}) is from gene expression measurements in
Bacillus Subtilis. The $p = 4088$ predictor variables are log-transformed
gene expressions and there is a response measuring the logarithm of the
production rate of riboflavin in Bacillus Subtilis. The data is kindly
provided by DSM (Switzerland). As the true variables are not known, we
consider a linear model with design matrix from real data and simulating a
sparse parameter vector $\beta$ as follows. In each simulation run, a new
parameter vector $\beta$ is created by either `uniform' or
`varying-strength' sampling. Under `uniform' sampling, $\abs{S}$ randomly
chosen components of $\beta$ are set to 1 and the remaining $p - \abs{S}$
components to 0. Under `varying-strength' sampling, $\abs{S}$ randomly
chosen components of $\beta$ are set to values $1,\ldots,\abs{S}$. The
error variance $\sigma^2$ is adjusted such that the signal to noise ratio
(SNR) is maintained at a desired level at each simulation run. We perform
50 simulations for each setting.

The sample-splitting is done such that the model is trained on a data set
of size $\lfloor (n-1)/2 \rfloor$ and the p-values are calculated
on the remaining data set. This slightly unbalanced scheme prevents us from
situations where the full model might be selected on the first
data set. Calculations of p-values would not be possible on the remaining
data in such a situation. We use a total of $B = 50$ sample-splits for each
simulation run. As in \citet{wasserman08highdim}, we compute p-values for all procedures using a normal approximation. Results are qualitatively similar when using a t-distribution instead. 

We compare the average number of true positives and the family-wise error
rate (FWER) for the single- and multi-split methods for all three simulation
settings (A)--(C) and vary in each the SNR to 0.25, 1, 4 and 16 (which
corresponds to population $R^2$ values of 0.2, 0.5, 0.8 and 0.94,
respectively). The number $|S|$ of relevant variables is either 5 or 10. As
initial variable selection or screening method $\tilde{S}$ we use three
approaches, which are all based on the Lasso \citep{tibshirani96lasso}. The
first one, denoted by $\tilde{S}_{fixed}$, uses the Lasso and selects those
$\lfloor n/6 \rfloor$ variables which appear most often in the
regularization path when varying the penalty parameter. The constant number of $\lfloor n/6 \rfloor$ variables is chosen, somewhat arbitrarily, to ensure a reasonably large set of selected coefficients on the one hand and to ensure, on the other hand,
that least squares estimation will work reasonably well on the second half of the data with sample size $\lfloor n/2 \rfloor$. While the choice seems to work well in practice and can be implemented very easily and efficiently, it is still slightly arbitrary. Avoiding any such choices of non-data adaptive tuning parameters, the second method,
$\tilde{S}_{cv}$, uses the Lasso with penalty parameter chosen by 10-fold
cross-validation and selecting the variables whose corresponding estimated
regression coefficients are different from zero. The third method,
$\tilde{S}_{adap}$, is the adaptive Lasso of \citet{zou06adalasso} where
regularization parameters are chosen based on 10-fold cross-validation with
the Lasso solution used as initial estimator for the adaptive Lasso. The
selected variables are again the ones whose corresponding estimated
regression parameters are different from zero.

Results are shown in Figures \ref{fig:n100p100} and \ref{fig:DSM} for both
the single-split method and the multi-split method with the default setting
$\gamma_{\min}=0.05$. Using the multi-split method, the average number of
true positives (the variables in $S$ which are selected) is typically
slightly increased while the FWER (the probability of including variables
in $N$) is reduced sharply. The single-split method has often a FWER above
the level $\alpha=0.05$ at which it is asymptotically controlled while for
the multi-split method the FWER is above the nominal level only in few
scenarios. 
The asymptotic control seems to give a
good control in finite sample settings with the multi-split method, maybe apart from the method $\tilde{S}_{fixed}$ on the very high-dimensional dataset (C).
The  single-split method, in contrast, selects in nearly all settings too many noise variables, exceeding the desired FWER sometimes substantially. This suggests that the asymptotic error control seems to work better for finite sample sizes for the multi-split method.  Even though the multi-split method is more
conservative than the single-split method (having substantially lower
FWER), the number of true discoveries is often increased. We note that for
data (\ref{sim-dsm}), with $p = 4088$, and in general for low SNR, the
number of true positives is low since we control the very stringent
family-wise error criterion at $\alpha = 0.05$ significance level. As an
alternative, controlling less conservative error measures is possible
and is discussed in Section~\ref{section:extensions}.

We also experimented with using the value of $Q_j(\gamma)$ directly as an
adjusted p-value, without the adaptive choice of $\gamma$ but using a fixed
value $\gamma=0.5$ instead, i.e.\ looking at twice the median value of all
p-values across multiple data splits, as suggested in a different context by \citet{vandewiel2009tpe}. The results were not as convincing as
for the adaptive choice and we recommend the adaptive version with
$\gamma_{\min}=0.05$ as a good default choice.

\begin{figure}
\includegraphics[width=0.32\textwidth]{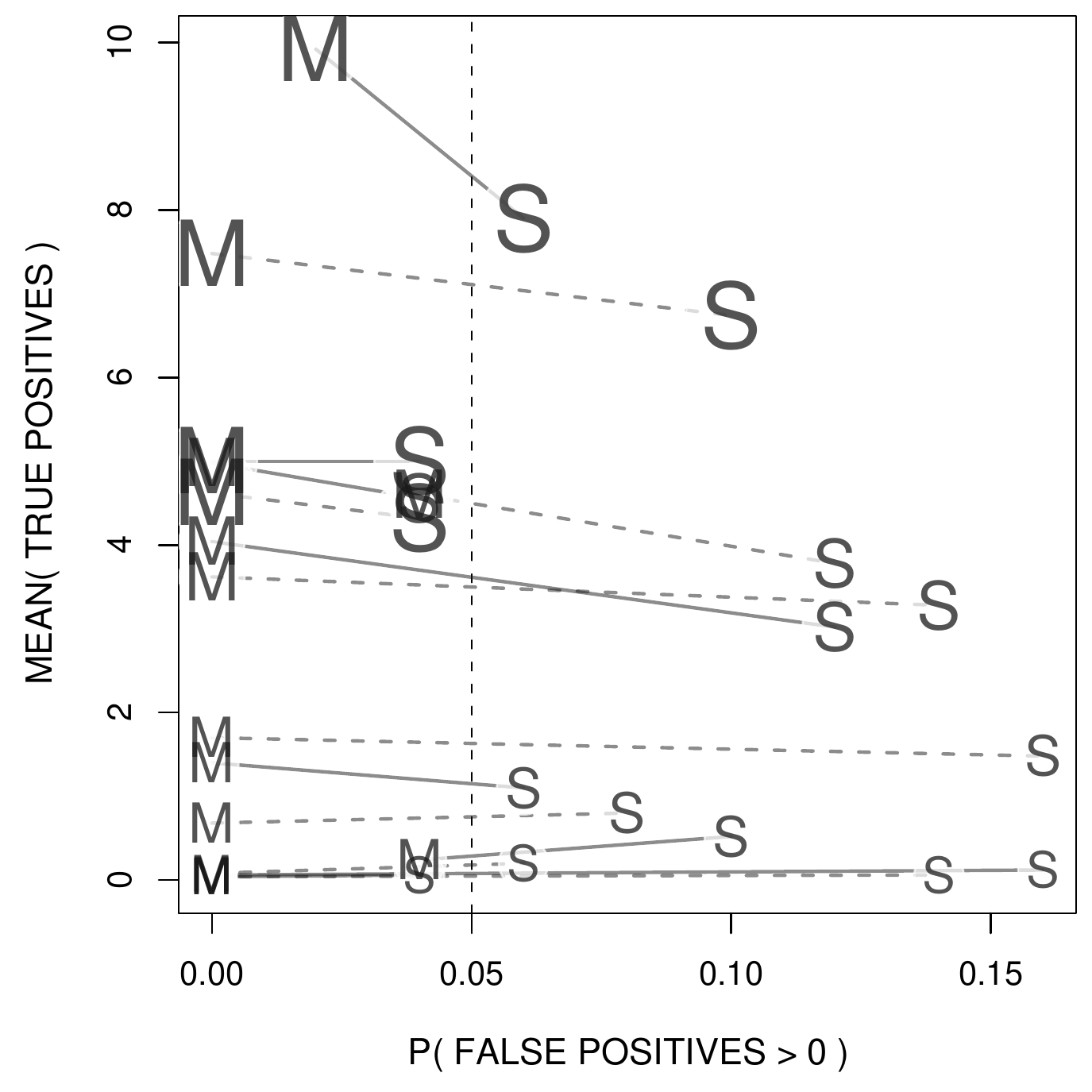}
\includegraphics[width=0.32\textwidth]{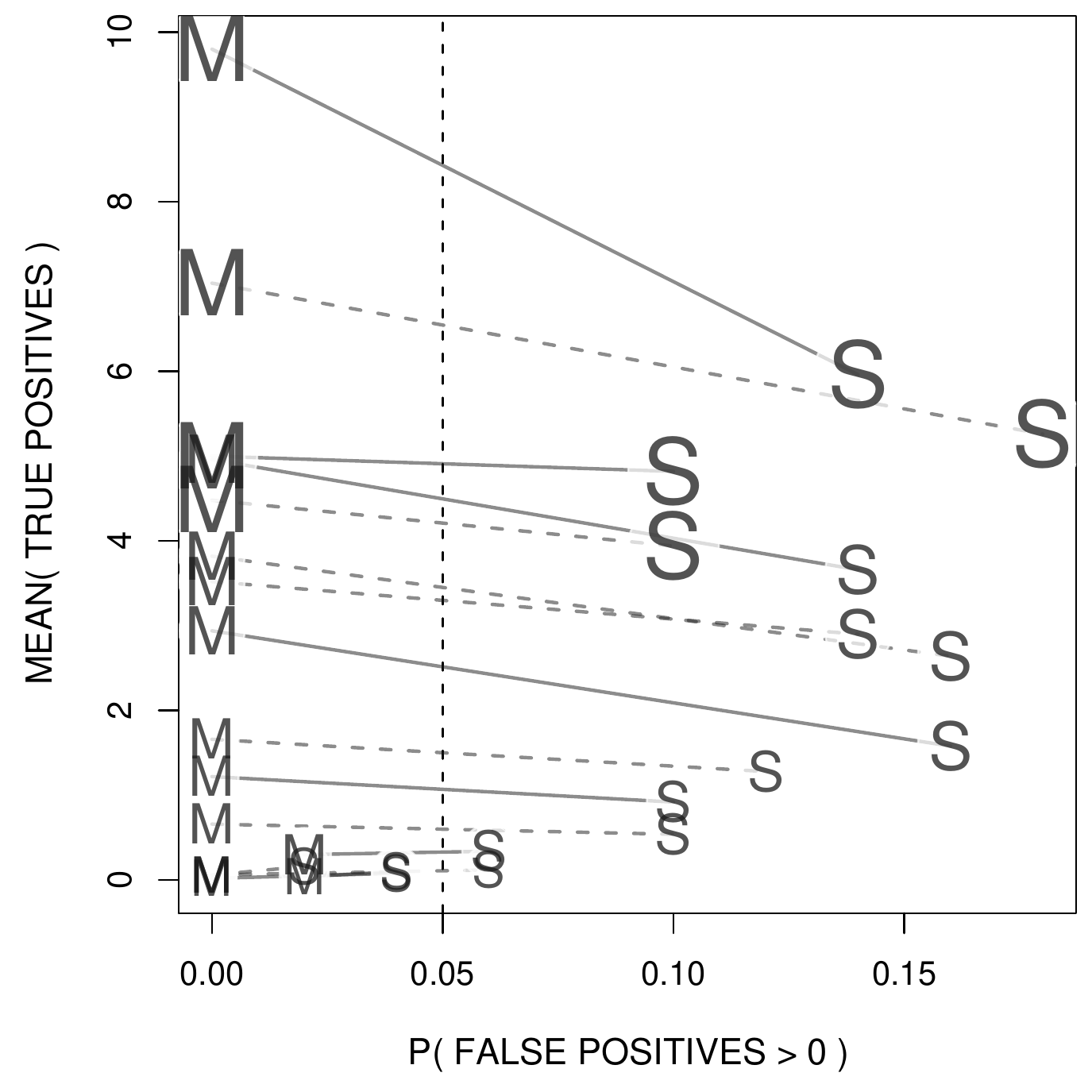}
\includegraphics[width=0.32\textwidth]{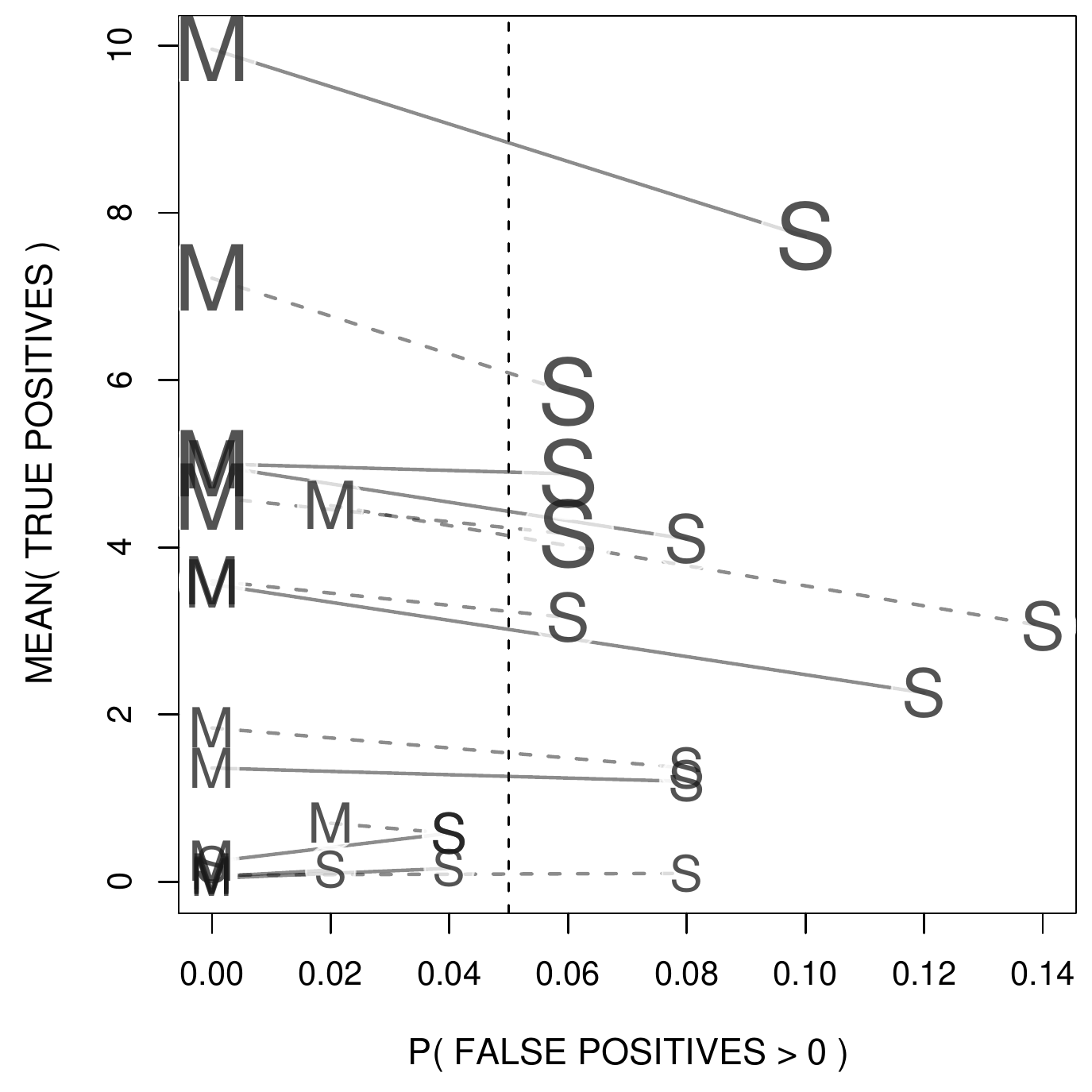}

\includegraphics[width=0.32\textwidth]{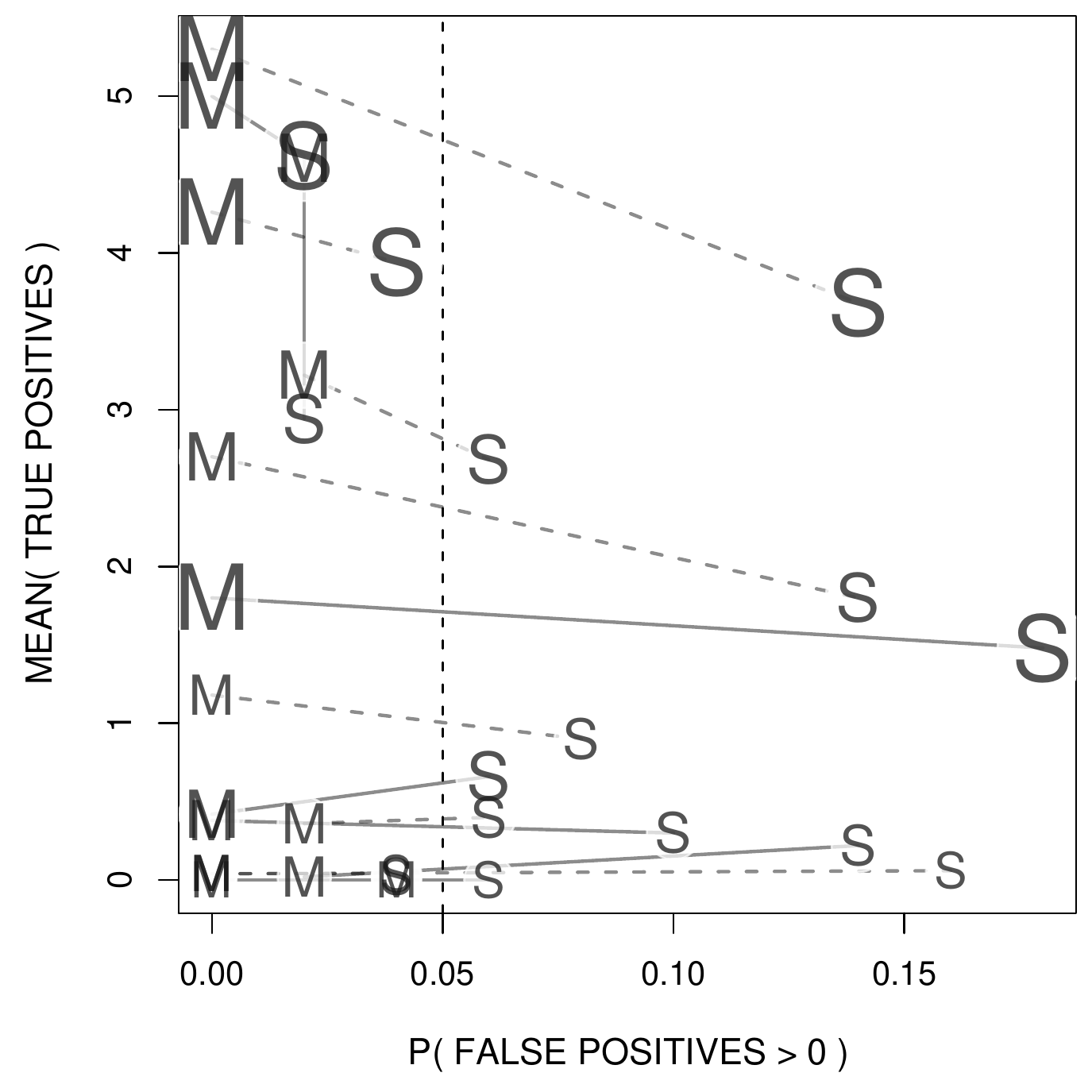}
\includegraphics[width=0.32\textwidth]{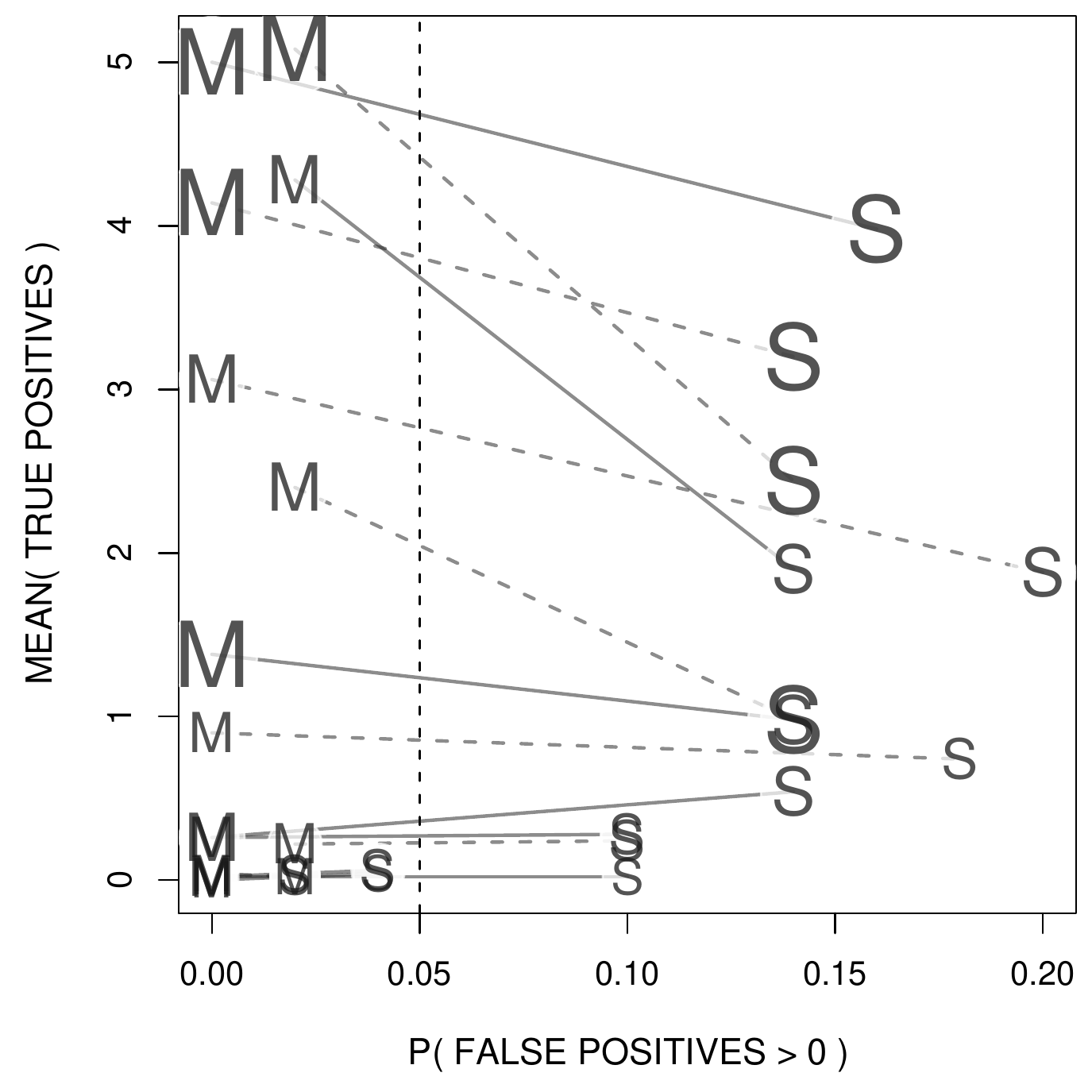}
\includegraphics[width=0.32\textwidth]{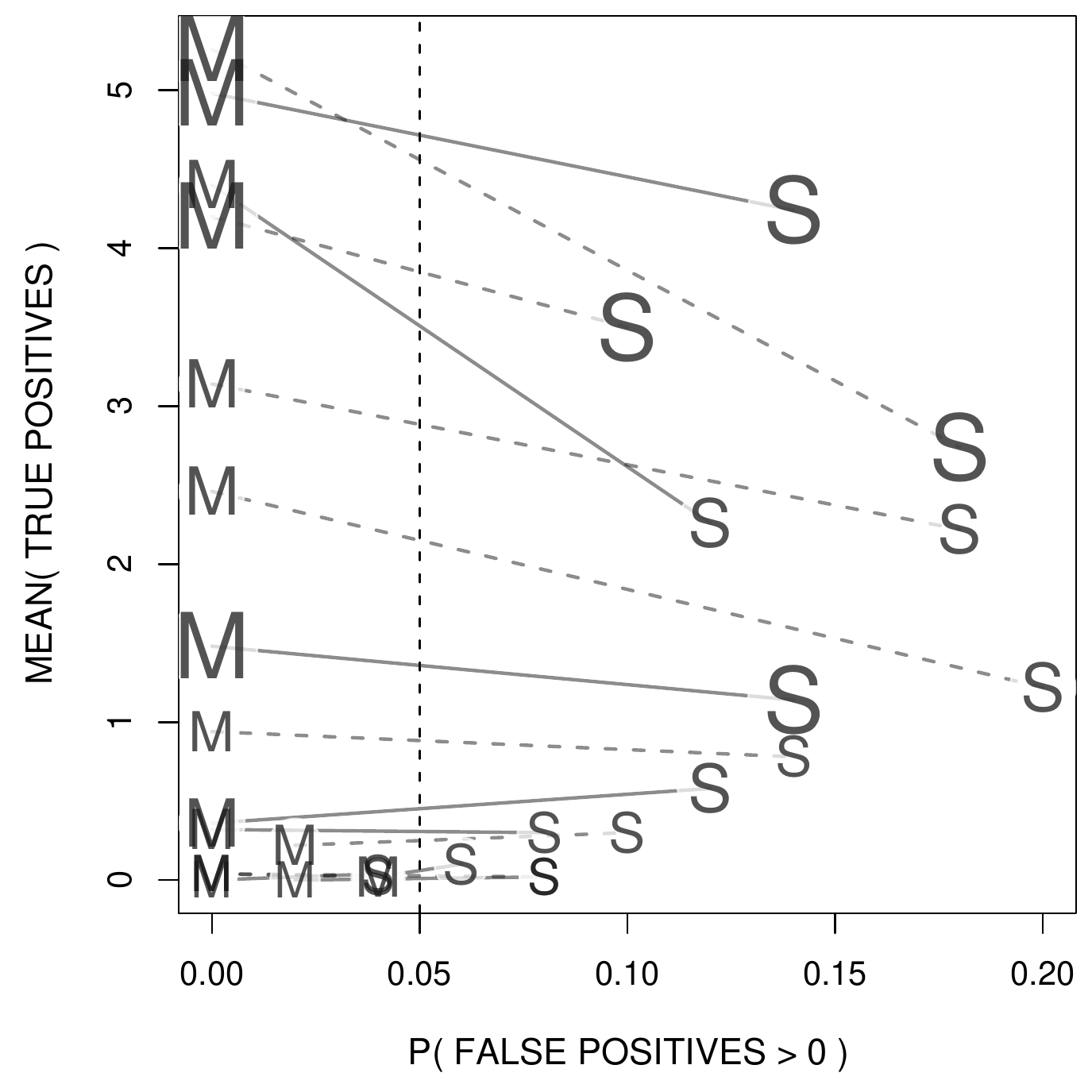}
\caption{\textit{\label{fig:n100p100} Simulation results for setting
    (\ref{sim-lowdim}) in the top and (\ref{sim-highdim}) in the bottom
    row. Average number of true positives vs.\ the family-wise error rate
    (FWER) for the single split method (`S') against the multi-split
    version (`M'). FWER is controlled (asymptotically) at $\alpha=0.05$ for
    both methods and this value is indicated by a broken vertical
    line. From left to right are results for $\tilde{S}_{fixed}$,
    $\tilde{S}_{cv}$ and $\tilde{S}_{adap}$. Results of a unique setting of
    SNR, sparsity and design are joined by a line, which is solid if the
    coefficients follow the `uniform' sampling and broken
    otherwise. Increasing SNR is indicated by increasing symbol size.}}
\end{figure}

\begin{figure}
\includegraphics[width=0.32\textwidth]{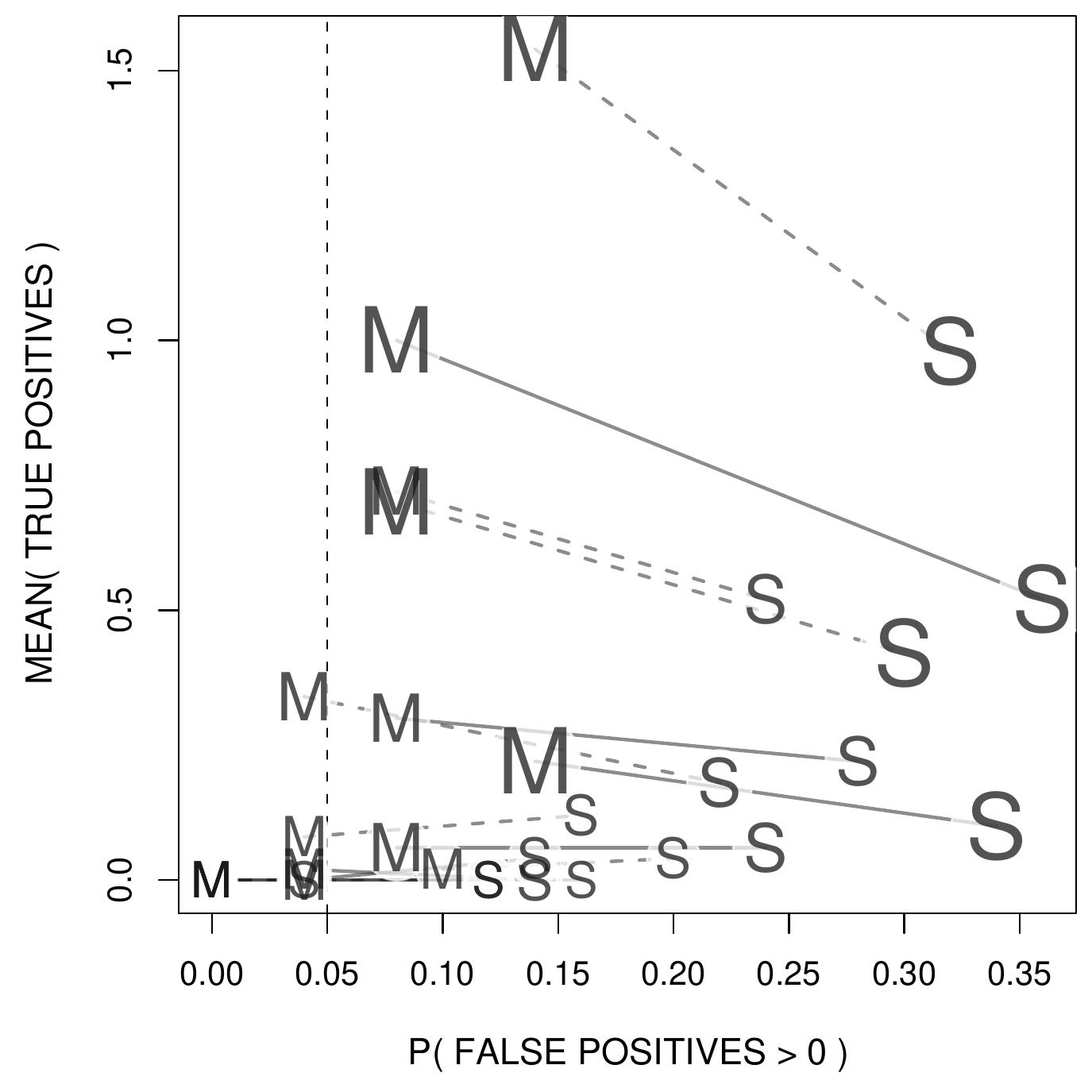}
\includegraphics[width=0.32\textwidth]{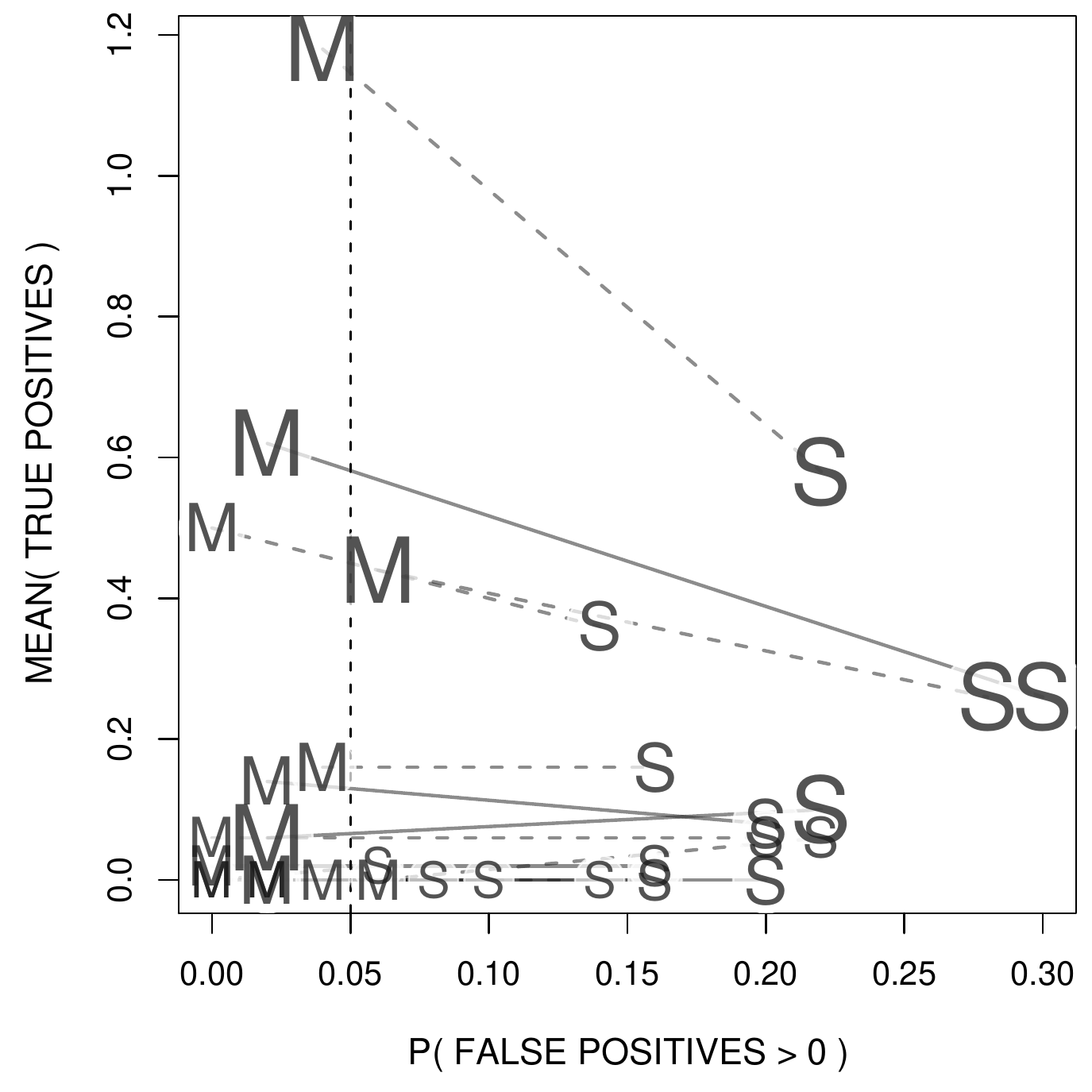}
\includegraphics[width=0.32\textwidth]{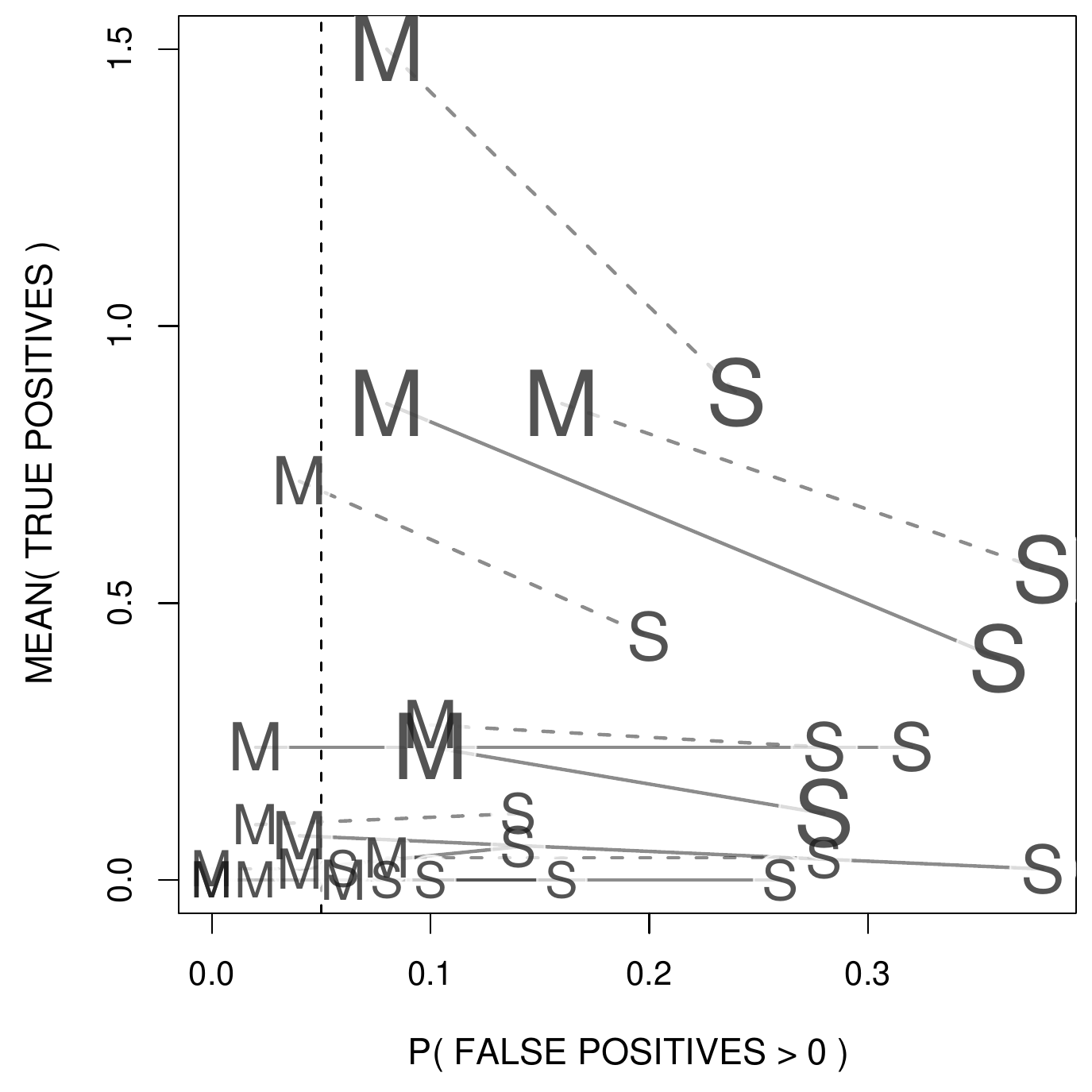}
\caption{\textit{ \label{fig:DSM} The results of simulation setup (C).}}
\end{figure}

\subsection{Comparisons with adaptive Lasso}
\label{sec:adaptivelasso}

Next, we compare the multi-split selector with the adaptive Lasso
\citep{zou06adalasso}. We have used the adaptive Lasso previously as a
variable selection method in our proposed multi-split method. The adaptive
Lasso is usually employed on its own. There are a few choices to make when
using the adaptive Lasso. We use the same choices as
previously. The initial estimator is obtained as the Lasso solution with a
10-fold CV-choice of the penalty parameter. The adaptive Lasso penalty is
also obtained by 10-fold CV.

Despite desirable asymptotic consistency properties \citep{huang2008als},
the adaptive Lasso does not offer error control in the same way as
Theorem~\ref{theo:fam-error-single} does for the multi-split method.  In
fact, the FWER (the probability of selecting at least one noise variable)
is very close to~1 with the adaptive Lasso in all the simulations we have
seen. In contrast, our multi-split method offers asymptotic control, which
was seen to be very well matched by the empirical FWER in the vicinity of
$\alpha=0.05$. Table \ref{COMPTABLEn100p200} shows the simulation results
for the multi-split method using $\tilde{S}_{adap}$ and the adaptive Lasso
on its own side by side for a simulation setting with $n = 100$, $p = 200$
and the same settings as in (A) and (B) otherwise. The adaptive Lasso
selects roughly 20 noise variables (out of $p=200$ variables), even though
the number of truly relevant variables is just 5 or 10. The average number
of false positives is at most 0.04 and often simply 0 with the proposed
multi-split method.

\begin{table}[h!]
\begin{center}
\begin{small}
\begin{tabular}{ccc|cc|cc|cc}
\multicolumn{3}{c|}{\bf }&\multicolumn{2}{c|}{\bf E( True Positives )} 
&\multicolumn{2}{c|}{\bf E( False Positives )}&\multicolumn{2}{c}{\bf  
P( False Positives $>$ 0 )}\\ \hline
Uniform&&&Multi&Adaptive&Multi&Adaptive&Multi&Adaptive\\
Sampling&$|S|$&SNR&Split&Lasso&Split&Lasso&Split&Lasso\\  \hline
NO&10&0.25&0.00&2.30&0&9.78&0&0.76\\
NO&10&1&0.58&6.32&0&20.00&0&1\\
NO&10&4&4.14&8.30&0&25.58&0&1\\
NO&10&16&7.20&9.42&0.02&30.10&0.02&1\\
YES&10&0.25&0.02&2.52&0&10.30&0&0.72\\
YES&10&1&0.10&7.46&0.02&21.70&0.02&1\\
YES&10&4&2.14&9.96&0&28.46&0&1\\
YES&10&16&9.92&10.00&0.04&30.66&0.04&1\\
NO&5&0.25&0.06&1.94&0&11.58&0&0.84\\
NO&5&1&1.50&3.86&0.02&19.86&0.02&1\\
NO&5&4&3.52&4.58&0.02&23.56&0.02&1\\
NO&5&16&4.40&4.98&0&27.26&0&1\\
YES&5&0.25&0.02&2.22&0&12.16&0&0.8\\
YES&5&1&0.82&4.64&0.02&22.18&0.02&1\\
YES&5&4&4.90&5.00&0&24.48&0&1\\
YES&5&16&5.00&5.00&0&28.06&0&1
\end{tabular}
\end{small}
\vspace{3mm}
\caption{\textit{\label{COMPTABLEn100p200} Comparing the multi-split method with
  CV-Lasso selection, $\tilde{S}_{adap}$, with the selection made when using
  the adaptive Lasso and a CV-choice of the involved penalty parameters for a 
  setting with $n = 100$ and $p = 200$.}}
\end{center}
\end{table}

There is clearly a price to pay for controlling the family-wise error
rate. Our proposed multi-split method detects on average less truly
relevant variables than the adaptive Lasso. For very low SNR, the
difference is most pronounced. The multi-split method selects in general
neither correct nor wrong variables for $\mbox{SNR}=0.25$, while the
adaptive Lasso averages between 2 to 3 correct selections, among 9-12 wrong
selections. Depending on the objectives of the study, one would prefer
either of the outcomes. For larger SNR, the multi-split method detects
almost as many truly important variables as the adaptive Lasso, while still
reducing the number of falsely selected variables from 20 or above to
roughly 0.

The multi-split method seems hence beneficial in settings where the cost of
making an erroneous selection is rather high. For example, expensive
follow-up experiments are usually required to validate results in
bio-medical applications and a stricter error control will place more of
the available resources into experiments which are likely to be successful.

\subsection{Motif regression}
\label{sec:realdata-motifregression}
We apply the multi-split method to a real data set about motif regression
\citep{conlon03motifregression}. For a total of $n = 287$ DNA segments we
have the binding intensity of a protein to each of the segments. These will
be our response values $Y_1, \ldots, Y_n$. Moreover, for $p = 195$
candidate words (`motifs') we have scores $x_{ij}$ which measure how well
the $j$th motif is represented in the $i$th DNA sequence. The motifs are
typically 5--15bp long candidates for the true binding site of the
protein. The hope is that the true binding site is in the list of
significant variables showing the strongest relationship between the motif
score and the binding intensity. Using a linear model with
$\tilde{S}_{adap}$, the multi-split method identifies one predictor
variable at the 5\% significance level. The single-split method is not able
to identify a single significant predictor. In view of the asymptotic error
control and the empirical results in Section~\ref{section:numerical} there
is substantial evidence that the selected variable corresponds to a true
binding site. For this specific application it seems desirable to pursue a
conservative approach with low FWER. As mentioned above, we could control
other, less conservative error measures as discussed in Section~\ref{section:extensions}.

\subsection{Comparison with standard low-dimensional FDR control}

\begin{figure}
\begin{center}
\includegraphics[width=0.99\textwidth]{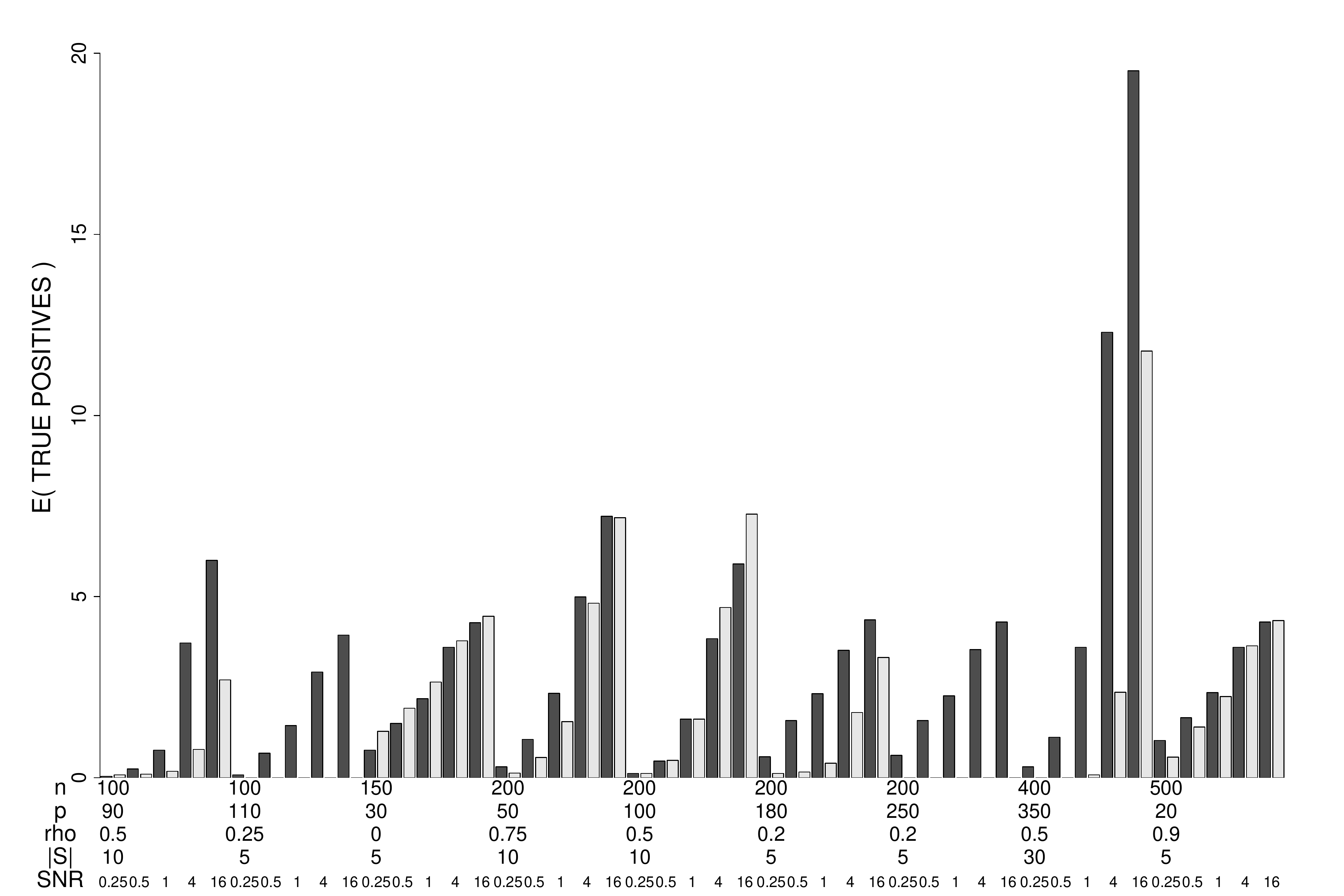}
\caption{\textit{ \label{fig:FDR} The results of FDR controlling simulations for the multi-split method (dark bar) and standard FDR control (light bar). The settings of $n,p,\rho,|S|$ and SNR are given below each simulation. The height of the bars corresponds to the average number of selected important variables. For $p>n$, the standard method breaks down and the corresponding bars are set to height 0.  }}
\end{center}
\end{figure}

We mentioned that control of FDR can be an attractive alternative to FWER if we expect 
a sizable number of rejections. Using the corrected p-values $P_1,\ldots,P_p$, a simple FDR-controlling procedure was derived in Section~\ref{section:FDR} and its asymptotic control of FDR was shown in Theorem~\ref{theo:fdr-all}. We now look empirically at the behavior of the resulting method and its power to detect truly interesting variables. Turning again to the simulation setting (\ref{sim-lowdim}), we vary the sample size $n$, the number of variables $p$, the signal to noise ratio SNR, the correlation $\rho$ between neighboring variables and the number $s$ of truly interesting variables.  

It was shown already above extensively that the multi-split method is preferable to the single-split method. Here, we are more interested in comparison to well understood traditional FDR controlling procedures. For $p<n$, the standard approach would be to compute the least squares estimator once for the full dataset. For each variable a p-value is obtained and the FDR controlling procedure as in (\ref{kFDR}) can be applied. This approach obviously breaks down for $p>n$. Our proposed approach can be applied both to low-dimensional ($p<n$) and high-dimensional ($p\ge n$) settings. 

In all settings, the empirical FDR of our method (not shown) is below $q = 0.05$
and often close to zero. Results regarding power are shown in Figure~\ref{fig:FDR} for control at $q=0.05$. 

It is maybe unexpected, but the multi-split method  tracks the power of the standard FDR controlling procedure quite closely for low-dimensional data $p<n$. In fact, the multi-split method is doing considerably better if $n/p$ is below, say, 1.5 or the correlation among the tests is large. An intuitive explanation for this behavior is that, as $p$ approaches $n$, the variance in each estimated coefficient vector under the OLS estimate  is increasing substantially. This in turn increases the variance of all OLS components $\hat{\beta}_j$, $j=1,\ldots,p$ and reduces the ability to select the truly important variables. The multi-split method, in contrast, trims the total number of variables to a substantially smaller number on one half of the samples and suffers then less from an increased variance in the estimated coefficients on the second half of the samples. Repeating this over multiple splits leads thus to  a surprisingly powerful variable selection procedure even for low-dimensional data. Nevertheless, we think  that the main application will be high-dimensional data, where the standard approach breaks down completely.

\section{Extensions}\label{section:extensions}
Due to the generic nature of our proposed methodology, extensions to any
situation where (asymptotically valid) p-values $\tilde{P}_j$ for
hypotheses $H_{0,j} \; (j = 1, \ldots, p)$ are available are
straightforward. An important class of examples are generalized linear
models (GLMs) or Gaussian Graphical Models. The dimensionality reduction
step would typically involve some form of shrinkage estimation. An example
for Gaussian Graphical Models would be the recently proposed `Graphical
Lasso' \citep{friedman2008sic}. The second step would rely on classical
(e.g.\ likelihood ratio) tests applied to the selected submodel, analogous
to the methodology proposed for linear regression.

In some settings, control of FWER at, say, $\alpha=0.05$ is too conservative. One can either resort to control of FDR, as alluded to above. 
Alternatively, FWER control can easily be adjusted to control the expected number of false
rejections. Take as an example the adjusted p-value $P_j$, defined in
(\ref{eq:quant-pval-opt}).  Variable $j$ is rejected if and only if $P_j\le
\alpha$. (For the following, assume that adjusted p-values, as defined in
(\ref{eq:pvaladj}), are not capped at~1. This is a technical detail only as
it does not modify the proposed FWER controlling procedure.)  Rejecting
variable $j$ if and only if $P_j\le\alpha$ controls FWER at level
$\alpha$. Instead, one can reject variables if and only if $P_j/K \le
\alpha$, where $K>1$ is a correction factor. Call the number of falsely
rejected variables $V$,
\[ V= \sum_{j\in N} 1\{ P_j /K \le\alpha \}. \] Then the expected number of
false positives is controlled at level $ \limsup_{n\rightarrow\infty} \E[V]
\le \alpha K .$ A proof of this follows directly from the proof of
Theorem~\ref{theo:fam-error-all}. Of course, we can equivalently set
$k=\alpha K$ and obtain a control $\limsup_{n\rightarrow\infty} \E[V]\le
k$. For example, setting $k=1$ offers a much less conservative error
control, if so desired, than control of the family-wise error rate.

\section{Discussion}
We proposed a multi-sample-split method for assigning statistical
significance and constructing conservative p-values for hypothesis testing
for high-dimensional problems where the number of predictor variables may
be much larger than sample size. Our method is an extension of the
single-split approach of \cite{wasserman08highdim} and is extended to false discovery rate (FDR) control. Combining the results of
multiple data-splits, based on quantiles as summary statistics, improves
reproducibility compared to the single-split method. The multi-split method
shares with the single-split method the property of asymptotic error
control and model selection consistency. We argue empirically that the
multi-split method usually selects much fewer false positives than the
single-split method while the number of true positives is slightly
increased. The main area of application will be high-dimensional data, where the number $p$ of predictor variables exceeds sample size $n$, as standard approaches rely on least-squares estimation and thus fail in this setting. It was, however, shown that the method is also an interesting alternative to standard FDR and FWER control in lower-dimensional settings as the proposed FDR control can be more powerful if $p$ is reasonably large but smaller than sample size $n$. The method is very generic and can be used for a broad spectrum
of error controlling procedures in multiple testing, including linear and
generalized linear models.

\appendix 
\section{Proofs}

\begin{proof}[Proof of Theorem \ref{theo:fam-error-single}]
  For technical reasons we define
  \begin{equation}\label{eq:defK}
    K_j^{(b)} = P_j^{(b)} 1\{S \subseteq \tilde{S}^{(b)}\} + 1\{S \not
    \subseteq \tilde{S}^{(b)}\}.
  \end{equation}
  $K_j^{(b)}$ are the adjusted p-values if the estimated active set
  contains the true active set. Otherwise, all p-values are set to
  1. Because of assumption (A\ref{screening-property}) and for fixed $B$,
  $\mathbb{P}[K_j^{(b)} = P_j^{(b)} \; \textrm{for all} \; b = 1, \ldots,
  B]$ on a set $A_n$ with $\mathbb{P}[A_n] \to 1$. Therefore, we can define
  all the quantities involving $P_j^{(b)}$ also with $K_j^{(b)}$, and it is
  sufficient to show under this slightly altered procedure that
  \[
    \Prob[\min_{j \in N} Q_j(\gamma) \leq \alpha] \leq \alpha.
  \]
  In particular we can omit here the limes superior.

  We also omit for the proofs the function $\min\{1,\cdot\}$ from the definitions of $Q_j(\gamma)$ and $P_j$ in (\ref{eq:quant-pval}) and (\ref{eq:quant-pval-opt}) respectively. The selected sets of variables are clearly unaffected and notation is simplifies considerably. 

  Define for $u\in(0,1)$ the quantity $\pi_j(u)$ as the fraction of bootstrap samples that yield $Kj^{(b)}$ less than or equal to $u$,
  \[
    \pi_j(u ) = \frac{1}{B} \sum_{b = 1}^{B} 1\big\{K_j^{(b)} 
    \leq u \big\}.
  \]
  Note that the events $\{Q_j(\gamma) \leq \alpha\}$ and $\{\pi_j(\alpha \gamma)
  \geq \gamma\}$ are equivalent. Hence,
  \begin{equation}\label{Qsum}
    \Prob\Big[ \min_{j \in N} Q_j(\gamma) \leq \alpha \Big]  \leq  
    \sum_{j \in N} \E\Big[1\big\{Q_j(\gamma) \leq \alpha \big\}\Big] 
     =  \sum_{j \in N} \E\Big[1\big\{\pi_j(\alpha \gamma) \geq \gamma
      \big\}\Big].
  \end{equation}
  Using a Markov inequality,
  \[
    \sum_{j \in N} \E\Big[1\big\{\pi_j(\alpha\gamma) \geq \gamma
      \big\}\Big] \leq \frac{1}{\gamma} \sum_{j \in N} \E[\pi_j(\alpha\gamma)].
  \]
  By definition of $\pi_j(\cdot)$,
  \[
    \frac{1}{\gamma} \sum_{j \in N} \E[\pi_j(\alpha \gamma)] = 
    \frac{1}{\gamma} \frac{1}{B} \sum_{b=1}^B \sum_{j \in N \cap 
      \tilde{S}^{(b)}} \E\Big[1\big\{K_j^{(b)} \leq \alpha \gamma \big\}\Big].
  \]
  Moreover, using the definition of $K_j^{(b)}$ in (\ref{eq:defK}),
  \[
    \E\Big[1\big\{K_j^{(b)} \leq \alpha \gamma \big\}\Big] \leq 
    \Prob\Big[P_j^{(b)} \leq \alpha \gamma \, \big| \, S \subseteq
    \tilde{S}^{(b)} \Big] = \frac{\alpha \gamma}{|\tilde{S}^{(b)}|}.
  \]
  This is a consequence of the uniform distribution of $\tilde{P}_j^{(b)}$
  given $S \subseteq \tilde{S}^{(b)}$. Summarizing these results we get
  \[
    \Prob\Big[ \min_{j \in N} Q_j(\gamma) \leq \alpha \Big] \leq   \frac{1}{\gamma} \frac{1}{B} \sum_{b=1}^B \E\Big[  \sum_{j \in N \cap 
      \tilde{S}^{(b)}} \frac{\alpha \gamma}{|\tilde{S}^{(b)}|} \Big] \leq \alpha,
  \]which completes the proof.
\end{proof}

\begin{proof}[Proof of Theorem \ref{theo:fam-error-all}]
  As in the proof of Theorem \ref{theo:fam-error-single} we will work with
  $K_j^{(b)}$ instead of $P_j^{(b)}$. Analogously, instead of
  $\tilde{P}_j^{(b)}$ we work with $\tilde{K}_j^{(b)}$.

  For any $\tilde{K}_j^{(b)}$ with $j \in N$ and $\alpha\in(0,1)$,
  \begin{equation}\label{use2}
    \E\leftb[ \frac{ 1\leftsmall\{ \tilde{K}_j^{(b)}\le\alpha\gamma\rightb\}}
      {\gamma} \rightb] \le \alpha.
  \end{equation}
  Furthermore, 
\[ 
\E\leftb[ \max_{j \in N}  \frac{ 1\leftsmall\{ K_j^{(b)}
      \le \alpha \gamma \rightsmall\}}{\gamma} \rightb] \le  \E\leftb[ \sum_{j \in N}  \frac{ 1\leftsmall\{ K_j^{(b)}
      \le \alpha \gamma \rightsmall\}}{\gamma} \rightb] 
\le  \E\leftb[ \sum_{j \in N\cap \tilde{S}^{(b)}}  \frac{ 1\leftsmall\{ K_j^{(b)}
      \le \alpha \gamma \rightsmall\}}{\gamma} \rightb] 
\]
and hence, with (\ref{use2}) and using the definition (\ref{eq:defK}) of $K_j^{(b)}$,
  \begin{equation}\label{use1}
    \E\leftb[ \max_{j \in N}  \frac{ 1\leftsmall\{ K_j^{(b)}
      \le \alpha \gamma \rightsmall\}}{\gamma} \rightb]  \le \E\leftb[ \sum_{j \in N\cap
      \tilde{S}^{(b)}} \frac{\alpha} {|\tilde{S}^{(b)}|}\rightb]  \le \alpha. 
  \end{equation}
  For a random variable $U$ taking values in $[0,1]$,
  \[
    \sup_{\gamma\in(\gamma_{\min},1)} \frac{ 1\leftsmall\{
    U\le\alpha\gamma\rightsmall\}} {\gamma} =
    \left\{ \begin{array}{cl} 0 &U \ge \alpha, \\
      \alpha/U & \alpha\gamma_{\min} \le U
      <\alpha, \\ 1/\gamma_{\min} & U <
      \alpha\gamma_{\min}. \end{array} \right.
  \]
  Moreover, if $U$ has a uniform distribution on $[0,1]$,
  \[ 
    \E\leftb[ \sup_{\gamma\in(\gamma_{\min},1)} \frac{ 1\leftsmall\{
      U\le\alpha\gamma\rightsmall\}} {\gamma} \rightb] =
    \int_{0}^{\alpha\gamma_{\min}} \gamma_{\min}^{-1} dx +
    \int_{\alpha\gamma_{\min}}^\alpha \alpha x^{-1} dx = \alpha (1 - \log
    \gamma_{\min}).
  \] 
  Hence, by using that $\tilde{K}_j^{(b)}$ has a uniform distribution on
  $[0,1]$ for all $j \in N$, conditional on $S \subseteq \tilde{S}^{(b)}$,
  \[ 
    \E\leftb[ \sup_{\gamma \in (\gamma_{\min},1)} \frac{1\leftsmall\{
      \tilde{K}_j^{(b)} \le \alpha \gamma \rightsmall\}}{\gamma} \rightb] 
   \leq \E\leftb[ \sup_{\gamma \in (\gamma_{\min},1)} \frac{1\leftsmall\{
      \tilde{K}_j^{(b)} \le \alpha \gamma \rightsmall\}}{\gamma} \, \big| \, S
    \subseteq \tilde{S}^{(b)} \rightb]
    = \alpha (1 - \log \gamma_{\min}).
  \]
  Analogously to (\ref{use1}), we can then deduce that
  \[
    \sum_{j \in N}\E\leftb[  \sup_{\gamma\in(\gamma_{\min},1)} \frac{ 1\leftsmall\{
      K_j^{(b)} \le\alpha\gamma\rightsmall\}} {\gamma} \rightb] \leq
    \alpha (1 - \log \gamma_{\min}).
  \]
  Averaging over all bootstrap samples yields
  \[ 
    \sum_{j \in N}\E\leftb[  \sup_{\gamma\in(\gamma_{\min},1)} \frac{
    \frac{1}{B} \sum_{b=1}^B 1\leftsmall\{ K_j^{(b)} /\gamma
    \le\alpha \rightsmall\} } {\gamma} \rightb] \le \alpha (1- \log \gamma_{\min}).
  \] 
  Using again a Markov inequality,
  \[  
    \sum_{j \in N}\E\leftb[ \sup_{\gamma\in(\gamma_{\min},1)}
    1\{\pi_j(\alpha \gamma) \geq \gamma\} \rightb]
    \le \alpha 
    (1-\log \gamma_{\min}),
  \]
  where we have used the same definition for $\pi_j(\cdot)$ as in the proof of
  Theorem \ref{theo:fam-error-single}.
  
  Since the events $\{Q_j(\gamma) \leq \alpha\}$ and $\{\pi_j(\alpha\gamma)
  \geq \gamma\}$ are equivalent, it follows that
  \[ 
    \sum_{j \in N}\Prob\leftb[  \inf_{\gamma\in (\gamma_{\min},1) }
    Q_j(\gamma) \le \alpha\rightb] \le \alpha (1-\log \gamma_{\min}),
  \] 
  implying that
  \[ 
    \sum_{j \in N}\Prob\leftb[  \inf_{\gamma\in (\gamma_{\min},1) }
    Q_j(\gamma)(1-\log \gamma_{\min}) \le \alpha\rightb] \le \alpha.
  \]
  Using the definition of $P_j$ in (\ref{eq:quant-pval-opt}), 
  \begin{equation}\label{eq:PBOUND}
     \sum_{j \in N}\Prob\leftb[ P_j \le \alpha\rightb] \le \alpha,
  \end{equation}
and thus, by the union bound,
\[     \Prob\leftb[ \min_{j \in N} P_j \le \alpha\rightb] \le \alpha,\]  which completes the proof.
\end{proof}

\begin{proof}[Proof of Theorem \ref{theo:fdr-all}]
We use  identical notation to the proof of Theorem 1.3 in \citet{benjamini01control}. An exception is that we use the value $q$ instead of $q/m$ in the FDR-controlling procedure since we are working with adjusted p-values. Let
\[ p_{ijk} = \Prob( \{ P_i \in [ (j-1) q , j q  ] \} \mbox{ and } C_k^{(i)}),   \]
where $C_k^{(i)}$ is the event that \emph{if} variable $i$ were rejected, \emph{then} $k-1$ other variables were also rejected. 
Now, as shown in equation (10) and then again in (28) in \citet{benjamini01control},
\[ \E(Q) = \sum_{i\in N} \sum_{k=1}^p \frac{1}{k} \sum_{j=1}^k p_{ijk}.\]
Using this result, we use in the beginning a similar argument to \citet{benjamini01control},
\begin{eqnarray} \E(Q) &=& \sum_{i\in N} \sum_{k=1}^p \frac{1}{k} \sum_{j=1}^k p_{ijk}  = \sum_{i\in N} \sum_{j=1}^p  \sum_{k=j}^p \frac{1}{k} p_{ijk} \nonumber \\ 
&\le & \sum_{i\in N} \sum_{j=1}^p  \sum_{k=j}^p \frac{1}{j} p_{ijk}  \le \sum_{i\in N} \sum_{j=1}^p \frac{1}{j} \sum_{k=1}^p  p_{ijk}  =  \sum_{j=1}^p \frac{1}{j}  \sum_{i\in N} \sum_{k=1}^p  p_{ijk} \label{ttt}
\end{eqnarray}
Let us denote 
\[ f(j) := \sum_{i\in N} \sum_{k=1}^p  p_{ijk}, \quad j=1,\ldots,p\] 
The last equation (\ref{ttt}) can then be rewritten as 
\begin{eqnarray} \E(Q) & \le & \sum_{j=1}^p  \frac{1}{j}f(j) =   f(1)  + \sum_{j=2}^p \frac{1}{j} \leftb( \sum_{j'=1}^j f(j') - \sum_{j'=1}^{j-1} f(j')  \rightb)  \\ & = &\sum_{j=1}^{p-1}  (\frac{1}{j} - \frac{1}{j+1})  \sum_{j'=1 }^j f(j')   +  \frac{1}{p} \sum_{j'=1 }^p f(j') \label{EQ} \end{eqnarray}
Note that, in analogy to (27) in \citet{benjamini01control},
\[ \sum_{k=1}^p  p_{ijk} = P\Big( \{ P_i \in [ (j-1) q , j q  ] \} \cap \big( \bigcup_{k}^p C_k^{(i)}\big ) \Big)  = P\Big(P_i \in [ (j-1) q , j q  ]\Big) \]
and hence 
\[ f(j) = \sum_{i\in N} \sum_{k=1}^p  p_{ijk}   =  \sum_{i\in N} P\Big(P_i \in [ (j-1) q , j q  ]\Big), \]
from which it follows by (\ref{eq:PBOUND}) in the proof of Theorem~\ref{theo:fam-error-all} that
\[ \sum_{j'=1}^j f(j') = \sum_{i\in N} P\Big(P_i \le jq \Big) \le jq. \]
Using this in (\ref{EQ}),
we obtain 
\begin{equation}\label{EQ} \E(Q)\; \le \; \sum_{j=1}^{p-1}  (\frac{1}{j} - \frac{1}{j+1})  jq   +  \frac{1}{p} pq = \Big( \sum_{j=1}^{p-1} \frac{1}{j(j+1)} j +1   \Big) q =  q \sum_{j=1}^p \frac{1}{j}, \end{equation}
which completes the proof.
\end{proof}

\begin{proof}[Proof of Corollary \ref{coro:model-consistency}]
  Because the single-split method is model selection consistent, it must
  hold that $\Prob[\max_{j\in S} \tilde{P}_j |\tilde{S}| \leq \alpha_n] \to
  1$ for $n \to \infty$. Using multiple data-splits, this property holds
  for each of the $B$ splits and hence $\Prob[ \max_{j\in S} \max_{b}
  \tilde{P}^{(b)}_j |\tilde{S}^{(b)}| \leq \alpha_n] \to 1$, which implies
  that, with probability converging to 1 for $n\to\infty$, the quantile
  $\max_{j\in S} Q_j(1)$ is bounded from above by $\alpha_n$. The maximum
  over all $j\in S$ of the adjusted p-values $P_j= (1-\log \gamma_{\min})
  \inf_{\gamma\in (\gamma_{\min},1)} Q_j(\gamma)$ is thus bounded from
  above by $(1-\log \gamma_{\min}) \alpha_n$, again with probability
  converging to 1 for $n\to\infty$.  
\end{proof}


\begin{thebibliography}{}

\bibitem[\protect\citeauthoryear{Bach}{Bach}{2008}]{bach08bootstrap}
Bach, F.~R. (2008).
\newblock Bolasso: {M}odel consistent {L}asso estimation through the bootstrap.
\newblock In {\em ICML '08: Proceedings of the 25th international conference on
  Machine learning}, New York, NY, USA, pp.\  33--40. ACM.

\bibitem[\protect\citeauthoryear{Benjamini and Hochberg}{Benjamini and
  Hochberg}{1995}]{benjamini95fdr}
Benjamini, Y. and Y.~Hochberg (1995).
\newblock Controlling the false discovery rate: A practical and powerful
  approach to multiple testing.
\newblock {\em Journal of the Royal Statistical Society Series B\/}~{\em
  57\/}, 289--300.

\bibitem[\protect\citeauthoryear{Benjamini and Yekutieli}{Benjamini and
  Yekutieli}{2001}]{benjamini01control}
Benjamini, Y. and D.~Yekutieli (2001).
\newblock The control of the false discovery rate in multiple testing under
  dependency.
\newblock {\em Annals of Statistics\/}~{\em 29}, 1165--1188.

\bibitem[\protect\citeauthoryear{Bickel, Ritov, and Tsybakov}{Bickel
  et~al.}{2008}]{bickel2008sal}
Bickel, P., Y.~Ritov, and A.~Tsybakov (2008).
\newblock Simultaneous analysis of {L}asso and {D}antzig selector.
\newblock {\em Annals of Statistics\/}.
\newblock To appear.


\bibitem[\protect\citeauthoryear{Blanchard and Roquain}{Blanchard and Roquain}{2008}]{blanchard2008tss}
Blanchard, G. and E.~Roquain (2008).
\newblock Two simple sufficient conditions for FDR control.
\newblock {\em Electronic Journal of Statistics\/}~{\em 2}, 963--992.

\bibitem[\protect\citeauthoryear{B\"uhlmann}{B\"uhlmann}{2006}]{buhlmann06high%
dimboost}
B\"uhlmann, P. (2006).
\newblock Boosting for high-dimensional linear models.
\newblock {\em Annals of Statistics\/}, 559--583.


\bibitem[\protect\citeauthoryear{Conlon, Liu, Lieb, and Liu}{Conlon
  et~al.}{2003}]{conlon03motifregression}
Conlon, E.~M., X.~S. Liu, J.~D. Lieb, and J.~S. Liu (2003).
\newblock Integrating regulatory motif discovery and genome-wide expression
  analysis.
\newblock {\em Proceedings of the National Academy of Science\/}~{\em 100},
  3339 -- 3344.

\bibitem[\protect\citeauthoryear{Fan and Lv}{Fan and Lv}{2008}]{fan08sis}
Fan, J. and J.~Lv (2008).
\newblock Sure independence screening for ultra-high dimensional feature space.
\newblock {\em Journal of the Royal Statistical Society Series B\/}~{\em
  70\/}, 849--911.

\bibitem[\protect\citeauthoryear{Friedman, Hastie, and Tibshirani}{Friedman
  et~al.}{2008}]{friedman2008sic}
Friedman, J., T.~Hastie, and R.~Tibshirani (2008).
\newblock {Sparse inverse covariance estimation with the graphical Lasso}.
\newblock {\em Biostatistics\/}~{\em 9\/}, 432.

\bibitem[\protect\citeauthoryear{Friedman}{Friedman}{2001}]{friedman01boosting}
Friedman, J.~H. (2001).
\newblock Greedy function approximation: A gradient boosting machine.
\newblock {\em The Annals of Statistics\/}~{\em 29}, 1189--1232.

\bibitem[\protect\citeauthoryear{Holm}{Holm}{1979}]{holm79simple}
Holm, S. (1979).
\newblock A simple sequentially rejective multiple test procedure.
\newblock {\em Scandinavian Journal of Statistics\/}~{\em 6}, 65--70.


\bibitem[\protect\citeauthoryear{Hothorn, Bretz, and Westfall}{Hothorn et~al.}{2008}]{hothorn08simultaneous}
Hothorn, T., F.~Bretz, and P.~Westfall (2008).
\newblock Simultaneous inference in general parametric models.
\newblock {\em Biometrical Journal\/}~{\em 50}, 346--363.



\bibitem[\protect\citeauthoryear{Huang, Ma, and Zhang}{Huang et~al.}{2008}]{huang2008als}
Huang, J. and Ma, S. and Zhang, C.H. (2008).
\newblock Adaptive lasso for sparse high-dimensional regression models.
\newblock {\em Statistica Sinica\/}~{\em 18}, 1603--1618.

\bibitem[\protect\citeauthoryear{Meinshausen}{Meinshausen}{2007}]{meinshausen0%
5relaxedlasso}
Meinshausen, N. (2007).
\newblock Relaxed {L}asso.
\newblock {\em Computational Statistics and Data Analysis\/}~{\em 52\/}, 374
  -- 393.

\bibitem[\protect\citeauthoryear{Meinshausen and B\"uhlmann}{Meinshausen and
  B\"uhlmann}{2006}]{meinshausen06lasso}
Meinshausen, N. and P.~B\"uhlmann (2006).
\newblock High-dimensional graphs and variable selection with the {L}asso.
\newblock {\em Annals of Statistics\/}~{\em 34\/}, 1436--1462.

\bibitem[\protect\citeauthoryear{Meinshausen and B\"uhlmann}{Meinshausen and
  B\"uhlmann}{2008}]{meinshausen08stability}
Meinshausen, N. and P.~B\"uhlmann (2008).
\newblock Stability selection.
\newblock Preprint.

\bibitem[\protect\citeauthoryear{Meinshausen and Yu}{Meinshausen and
  Yu}{2009}]{meinshausen06lassorecovery2}
Meinshausen, N. and B.~Yu (2009). 
\newblock Lasso-type recovery of sparse representations for
high-dimensional data.
\newblock {\em Annals of Statistics\/}~{\em 37\/}, 246--270.

\bibitem[\protect\citeauthoryear{Tibshirani}{Tibshirani}{1996}]{tibshirani96la%
sso}
Tibshirani, R. (1996).
\newblock Regression shrinkage and selection via the {L}asso.
\newblock {\em Journal of the Royal Statistical Society Series B\/}~{\em
  58\/}, 267--288.

\bibitem[\protect\citeauthoryear{Tropp and Gilbert}{Tropp and
  Gilbert}{2007}]{tropp07omp}
Tropp, J. and A.~Gilbert (2007).
\newblock Signal recovery from random measurements via orthogonal matching
  pursuit.
\newblock {\em IEEE Transactions on Information Theory\/}~{\em 53\/}(12), 4655
  -- 4666.

\bibitem[\protect\citeauthoryear{van~de Geer}{van~de
  Geer}{2008}]{geer07glmlasso}
van~de Geer, S. (2008).
\newblock High-dimensional generalized linear models and the {L}asso.
\newblock {\em Annals of Statistics\/}~{\em 36\/}, 614--645.


\bibitem[\protect\citeauthoryear{van~de Wiel, Berkhof and Wieringen}{van~de Wiel et~al.}{2009}]{vandewiel2009tpe}
van~de Wiel, M. and J.~Berkhof and W.~van Wieringen (2009).
\newblock Testing the prediction error difference between 2 predictors.
\newblock {\em Biostatistics\/}.
\newblock To appear.


\bibitem[\protect\citeauthoryear{Wasserman and Roeder}{Wasserman and
  Roeder}{2008}]{wasserman08highdim}
Wasserman, L. and K.~Roeder (2008).
\newblock High dimensional variable selection.
\newblock {\em Annals of Statistics\/}.
\newblock To appear.

\bibitem[\protect\citeauthoryear{Zhang and Huang}{Zhang and
  Huang}{2008}]{zhang07lasso}
Zhang, C.-H. and J.~Huang (2008).
\newblock The sparsity and bias of the {L}asso selection in high-dimensional
  linear regression.
\newblock {\em Annals of Statistics\/}~{\em 36\/}, 1567--1594.

\bibitem[\protect\citeauthoryear{Zhao and Yu}{Zhao and
  Yu}{2006}]{zhao06consistency}
Zhao, P. and B.~Yu (2006).
\newblock On model selection consistency of {L}asso.
\newblock {\em Journal of Machine Learning Research\/}~{\em 7}, 2541--2563.

\bibitem[\protect\citeauthoryear{Zou}{Zou}{2006}]{zou06adalasso}
Zou, H. (2006).
\newblock The adaptive {L}asso and its oracle properties.
\newblock {\em Journal of the American Statistical Association\/}~{\em
  101\/}, 1418--1429.

\end{thebibliography}
\end{document}